\documentclass[onecolumn]{IEEEtran}
\usepackage[cmex10]{amsmath} 
\usepackage{array}
\usepackage{threeparttable}
\usepackage{soul}
\usepackage{multirow}
\usepackage{hyperref}

\usepackage{amsmath,amsthm,mathtools}
\usepackage{amssymb}
\usepackage{mathrsfs}

\usepackage{multirow}
\usepackage{array}

\usepackage{amsfonts}
\usepackage{graphicx}
\usepackage[linesnumbered]{algorithm2e}

\usepackage{color, verbatim}
\usepackage[usenames,dvipsnames]{xcolor}
\usepackage{cleveref}

\usepackage{epsfig}
\usepackage{subfig}
\usepackage{epstopdf}
\usepackage{xargs}
\usepackage{hhline}

\newcommand\nc\newcommand
\nc{\bb}[1]{\mathbb{#1}}
\renewcommand{\cal}[1]{\mathcal{#1}}
\renewcommand{\bf}[1]{\mathbf{#1}}

\DeclarePairedDelimiter{\set}{\lbrace}{\rbrace}
\DeclarePairedDelimiter{\br}{\lparen}{\rparen}
\DeclarePairedDelimiter{\brac}{\lbrack}{\rbrack}
\DeclarePairedDelimiter{\abs}{\lvert}{\rvert}

\nc\bfa{{\bf{a}}}\nc\bfA{{\boldsymbol A}}\nc\cA{{\cal A}} \nc\fA[1]{A\br*{#1}} \nc\fa[1]{a\br*{#1}}  \nc\rmA{\mathrm{A}} \nc\rma{\mathrm{a}}
\nc\bfb{{\bf{b}}}\nc\bfB{{\boldsymbol B}}\nc\cB{{\cal B}} \nc\fB[1]{B\br*{#1}} \nc\fb[1]{b\br*{#1}}  \nc\rmB{\mathrm{B}} \nc\rmb{\mathrm{b}}
\nc\bfc{{\bf{c}}}\nc\bfC{{\boldsymbol C}}\nc\cC{{\cal C}} \nc\fC[1]{C\br*{#1}} \nc\fc[1]{c\br*{#1}}  \nc\rmC{\mathrm{C}} \nc\rmc{\mathrm{c}}
\nc\bfd{{\bf{d}}}\nc\bfD{{\boldsymbol D}}\nc\cD{{\cal D}} \nc\fD[1]{D\br*{#1}} \nc\fd[1]{d\br*{#1}}  \nc\rmD{\mathrm{D}} \nc\rmd{\mathrm{d}}
\nc\bfe{{\bf{e}}}\nc\bfE{{\boldsymbol E}}\nc\cE{{\cal E}} \nc\fE[1]{E\br*{#1}} \nc\fe[1]{e\br*{#1}}  \nc\rmE{\mathrm{E}} \nc\rme{\mathrm{e}}
\nc\bff{{\bf{f}}}\nc\bfF{{\boldsymbol F}}\nc\cF{{\cal F}} \nc\fF[1]{F\br*{#1}} \nc\ff[1]{f\br*{#1}}  \nc\rmF{\mathrm{F}} \nc\rmf{\mathrm{f}}
\nc\bfg{{\bf{g}}}\nc\bfG{{\boldsymbol G}}\nc\cG{{\cal G}} \nc\fG[1]{G\br*{#1}} \nc\fg[1]{g\br*{#1}}  \nc\rmG{\mathrm{G}} \nc\rmg{\mathrm{g}}
\nc\bfh{{\bf{h}}}\nc\bfH{{\boldsymbol H}}\nc\cH{{\cal H}} \nc\fH[1]{H\br*{#1}} \nc\fh[1]{h\br*{#1}}  \nc\rmH{\mathrm{H}} \nc\rmh{\mathrm{h}}
\nc\bfi{{\bf{i}}}\nc\bfI{{\boldsymbol I}}\nc\cI{{\cal I}} \nc\fI[1]{I\br*{#1}} \nc\rmI{\mathrm{I}} \nc\rmi{\mathrm{i}}
\nc\bfj{{\bf{j}}}\nc\bfJ{{\boldsymbol J}}\nc\cJ{{\cal J}} \nc\fJ[1]{J\br*{#1}} \nc\fj[1]{j\br*{#1}} \nc\rmJ{\mathrm{J}} \nc\rmj{\mathrm{j}}
\nc\bfk{{\bf{k}}}\nc\bfK{{\boldsymbol K}}\nc\cK{{\cal K}} \nc\fK[1]{K\br*{#1}} \nc\fk[1]{k\br*{#1}} \nc\rmK{\mathrm{K}} \nc\rmk{\mathrm{k}}
\nc\bfl{{\bf{l}}}\nc\bfL{{\boldsymbol L}}\nc\cL{{\cal L}} \nc\fL[1]{L\br*{#1}} \nc\fl[1]{l\br*{#1}} \nc\rmL{\mathrm{L}} \nc\rml{\mathrm{l}}
\nc\bfm{{\bf{m}}}\nc\bfM{{\boldsymbol M}}\nc\cM{{\cal M}} \nc\fM[1]{M\br*{#1}} \nc\fm[1]{m\br*{#1}} \nc\rmM{\mathrm{M}} \nc\rmm{\mathrm{m}}
\nc\bfn{{\bf{n}}}\nc\bfN{{\boldsymbol N}}\nc\cN{{\cal N}} \nc\fN[1]{N\br*{#1}} \nc\fn[1]{n\br*{#1}} \nc\rmN{\mathrm{N}} \nc\rmn{\mathrm{n}}
\nc\bfo{{\bf{o}}}\nc\bfO{{\boldsymbol O}}\nc\cO{{\cal O}} \nc\fO[1]{O\br*{#1}} \nc\fo[1]{o\br*{#1}} \nc\rmO{\mathrm{O}} \nc\rmo{\mathrm{o}}
\nc\bfp{{\bf{p}}}\nc\bfP{{\boldsymbol P}}\nc\cP{{\cal P}} \nc\fP[1]{P\br*{#1}} \nc\fp[1]{p\br*{#1}} \nc\rmP{\mathrm{P}} \nc\rmp{\mathrm{p}}
\nc\bfq{{\bf{q}}}\nc\bfQ{{\boldsymbol Q}}\nc\cQ{{\cal Q}} \nc\fQ[1]{Q\br*{#1}} \nc\fq[1]{q\br*{#1}} \nc\rmQ{\mathrm{Q}} \nc\rmq{\mathrm{q}}
\nc\bfr{{\bf{r}}}\nc\bfR{{\boldsymbol R}}\nc\cR{{\cal R}} \nc\fR[1]{R\br*{#1}} \nc\fr[1]{r\br*{#1}} \nc\rmR{\mathrm{R}} \nc\rmr{\mathrm{r}}
\nc\bfs{{\bf{s}}}\nc\bfS{{\boldsymbol S}}\nc\cS{{\cal S}} \nc\fS[1]{S\br*{#1}} \nc\fs[1]{s\br*{#1}} \nc\rmS{\mathrm{S}} \nc\rms{\mathrm{s}}
\nc\bft{{\bf{t}}}\nc\bfT{{\boldsymbol T}}\nc\cT{{\cal T}} \nc\fT[1]{T\br*{#1}} \nc\ft[1]{t\br*{#1}} \nc\rmT{\mathrm{T}} \nc\rmt{\mathrm{t}}
\nc\bfu{{\bf{u}}}\nc\bfU{{\boldsymbol U}}\nc\cU{{\cal U}} \nc\fU[1]{U\br*{#1}} \nc\fu[1]{u\br*{#1}} \nc\rmU{\mathrm{U}} \nc\rmu{\mathrm{u}}
\nc\bfv{{\bf{v}}}\nc\bfV{{\boldsymbol V}}\nc\cV{{\cal V}} \nc\fV[1]{V\br*{#1}} \nc\fv[1]{v\br*{#1}} \nc\rmV{\mathrm{V}} \nc\rmv{\mathrm{v}}
\nc\bfw{{\bf{w}}}\nc\bfW{{\boldsymbol W}}\nc\cW{{\cal W}} \nc\fW[1]{W\br*{#1}} \nc\fw[1]{w\br*{#1}} \nc\rmW{\mathrm{W}} \nc\rmw{\mathrm{w}}
\nc\bfx{{\bf{x}}}\nc\bfX{{\boldsymbol X}}\nc\cX{{\cal X}} \nc\fX[1]{X\br*{#1}} \nc\fx[1]{x\br*{#1}} \nc\rmX{\mathrm{X}} \nc\rmx{\mathrm{x}}
\nc\bfy{{\bf{y}}}\nc\bfY{{\boldsymbol Y}}\nc\cY{{\cal Y}} \nc\fY[1]{Y\br*{#1}} \nc\fy[1]{y\br*{#1}} \nc\rmY{\mathrm{Y}} \nc\rmy{\mathrm{y}}
\nc\bfz{{\bf{z}}}\nc\bfZ{{\boldsymbol Z}}\nc\cZ{{\cal Z}} \nc\fZ[1]{Z\br*{#1}} \nc\fz[1]{z\br*{#1}} \nc\rmZ{\mathrm{Z}} \nc\rmz{\mathrm{z}}

\DeclareMathOperator{\Rank}{rank}
\DeclareMathOperator{\spn}{span}

\nc\defeq{\coloneqq}

\newcommand{\rank}[1]{\Rank\br*{#1}}

\DeclarePairedDelimiterX\Set[1]\{\}{#1}

\newcommand\F{{\mathbb F}}

\newcommand\integers{{\mathbb Z}}

\newtheorem{theorem}{Theorem}
\newtheorem*{theorem*}{Theorem}

\crefname{definition}{defn.}{defns}
\newtheorem{lemma}[theorem]{Lemma}
\newtheorem*{lemma*}{Lemma}

\newtheorem{corollary}[theorem]{Corollary}
\newtheorem{corollary*}[theorem]{Corollary}

\newtheorem*{remark*}{Remark}

\theoremstyle{definition}
\newtheorem{property}{Properties}
\crefname{Appendix}{Appendix}{Appendices}

\begin{document}
\sloppy

\title{Linear Programming Approximations for \\Index Coding}
\author{\IEEEauthorblockN{Abhishek Agarwal}\hspace*{1cm}
\IEEEauthorblockN{Larkin Flodin}\hspace*{1cm}
\IEEEauthorblockN{Arya Mazumdar}
}

\allowdisplaybreaks
\sloppy
\maketitle

{\renewcommand{\thefootnote}{}\footnotetext{

\vspace{-.2in}
 
\noindent\rule{1.5in}{.4pt}

College of Information and Computer Sciences, University of Massachusetts Amherst. 
\texttt{\{abhiag,lflodin,arya\}@cs.umass.edu}. This work is supported in part by an NSF CAREER award CCF 1642658 and NSF award CCF 1618512. A part of the paper was presented in the Network Coding and Applications (NetCod 2016), IEEE GLOBECOM Workshop, Dec 2016 \cite{agarwal2016local}.
}
\renewcommand{\thefootnote}{\arabic{footnote}}
\setcounter{footnote}{0}
\newcommand{\VC}{\textrm{VC}}
\newcommand{\FVC}{\textrm{FVC}}
\newcommand{\CP}{\textrm{CP}}
\newcommand{\FCP}{\textrm{FCP}}
\newcommand{\CC}{\textrm{CC}}
\newcommand{\FCC}{\textrm{FCC}}
\renewcommand{\Cap}{\textrm{Cap}}
\newcommand{\Ind}{\textrm{Ind}}
\newcommand{\MM}{\textrm{MM}}
\newcommand{\FMM}{\textrm{FMM}}
\newcommand{\G}{\mathcal{G}}
\newcommand{\MAIS}{\textrm{MAIS}}

\begin{abstract}
Index coding, a source coding problem over broadcast channels, has been a subject of both theoretical and practical interest since its introduction (by Birk and Kol, 1998). In short, the problem can be defined as follows: there is an input $P \triangleq (p_1, \dots, p_n)$, a set of $n$ clients who each desire a single entry $p_i$ of the input, and a broadcaster whose goal is to send as few messages as possible to all clients so that each one can recover its desired entry. Additionally, each client has some predetermined ``side information,'' corresponding to certain entries of the input $P$, which we represent as the ``side information graph'' $\G$. The graph $\G$ has a vertex $v_i$ for client $i$ and a directed edge $(v_i, v_j)$ indicating that client $i$ knows the $j$th entry of the input. Given a fixed side information graph $\G$, we are interested in determining or approximating the ``broadcast rate'' of index coding on the graph, i.e. the least number of messages the broadcaster can transmit so that every client recovers its desired information. The complexity of determining this broadcast rate in the most general case is open, and the best known approximations are barely better than the trivial $O(n)$-approximation corresponding to sending each client their information directly without performing any coding.

Using index coding schemes based on linear programs (LPs), we take a two-pronged approach to approximating the broadcast rate. First, extending earlier work on planar graphs, we focus on approximating the broadcast rate for special graph families such as graphs with small chromatic number and disk graphs. In certain cases, we are able to show that simple LP-based schemes give constant-factor approximations of the broadcast rate, which seem extremely difficult to obtain in the general case. Second, we provide several LP-based schemes for the general case which are not constant-factor approximations, but which strictly improve on the best-known schemes. These can be viewed as both a strengthening of the constant-factor approximations proven for special graph families (as these schemes strictly improve on those which we prove are good approximations), as well as another tool that can be used either in practice or in future theoretical analyses.
\end{abstract}


\section{Introduction}
\label{sec:intro}

Index coding is a particular form of network coding that was first introduced by Birk and Kol \cite{birk1998informed}, and has since been shown to be in some sense as difficult as any other network coding problem \cite{DBLP:journals/tit/RouayhebSG10}. It is a multiuser communication problem in which a broadcaster aims to transmit data to many users. While the users are unable to communicate amongst themselves, some of them already possess data desired by other users, which we call the ``side information.'' The goal is then to design transmission schemes for the broadcaster and corresponding decoding schemes for the users that exploit this side information in order to get each user their desired data in a minimum number of broadcaster transmissions.

More formally, we have a set $C = \set{1, 2, \dotsc, n}$ of clients which we refer to simply by number, and each client desires the corresponding message from the set $P = \set{p_1, p_2, \dotsc, p_n}$, where each message $p_i$ belongs to an alphabet $\Sigma$ with $|\Sigma| =q$.  Additionally, each client has some side information $\Gamma_i \subseteq P$. We define the (directed) side information graph of the index coding instance to be the graph $\G$ with vertices $v_1, v_2, \dotsc, v_n$ corresponding to clients, and edges $(v_i, v_j)$ whenever $p_j \in \Gamma_i$.  Then the goal is for a broadcaster  to transmit $l$ messages, each belonging to $\Sigma$, simultaneously to all clients so that every client $i$ can reconstruct $p_i$ as a function of $\Gamma_i$ and the $l$ messages sent by the broadcaster.

Specifically, if there exists an encoding function $f: \Sigma^n \rightarrow \Sigma^l,$ and decoding functions $g_i: \Sigma^{l} \times \Sigma^{|\Gamma_i|} \rightarrow \Sigma, i = 1, 2,\dotsc, n,$ such that $g_i(f(p_1, p_2, \dotsc, p_n), \Gamma_i) = p_i$ for each $i$, 
then we say this is a solution to the index coding problem on $\G$ in $l$ rounds. The minimal number of rounds needed to obtain a solution also depends on $q$, the size of the alphabet. We define $\Ind_q(\G)$ to be the minimum number of rounds $l$ such that a solution exists on $G$ in $l$ rounds over an alphabet $\Sigma$ of size $q$. We then define the index coding rate or the broadcast rate of the graph $\G$ as
\begin{equation}
\Ind(\G) = \inf_{q\ge 2} \Ind_q(\G).
\end{equation}

Some special types of index coding scheme require attention before we continue further. Suppose $\Sigma = \F^m$ for some finite field $\F$ and the encoding function is linear over $\F$. If $m=1$ and the broadcaster sends only linear combinations of the messages $p_i \in \F$, the message-sending scheme is called \textbf{scalar linear}. For $m >1$, if the broadcaster is allowed to break up the messages in $\F^m$ into smaller packets in $\F$ and transmit linear combinations of the packets, the scheme is called \textbf{vector linear}. To be more precise, for scalar linear schemes, the encoding function consists of $l$ different functions $f_j:\F^n \rightarrow \F, j =1, 2, \dotsc, l$, where each function is an $\mathbb{F}$-linear combination of the arguments. For vector linear schemes,  the encoding function consists of $ml$ different functions $f_j:\F^{mn} \rightarrow \F, j =1, 2, \dotsc, ml$, where each function is an $\mathbb{F}$-linear combination of the arguments.
All scalar linear schemes are also vector linear schemes. If a scheme is not vector linear, it is called \textbf{nonlinear}.
 In this paper we will focus on the quality of solutions relative to the best possible nonlinear scheme, although all schemes we provide are vector linear.

\subsection{Related Work}

Without any restriction on the graph $\G$ or the encoding function, no bounded time algorithm is known for finding $\Ind(\G)$ exactly, as little is understood about the speed at which the rates converge (therefore even an exponential-time algorithm to estimate $\Ind(\G)$ is of interest). This is in contrast to the scalar linear case with fixed alphabet size, in which the broadcast rate is known to be equal to another graph parameter called ``minrank,'' and finding this quantity exactly is known to be in NP \cite{bar2011index}. The best known approximation factor in general is $O(n \frac{\log \log n}{\log n})$ (i.e. the scheme returned by the algorithm has rate at most a multiplicative factor of $O(n \frac{\log \log n}{\log n})$ larger than $\Ind(\G)$) \cite{DBLP:journals/corr/abs-1004-1379}, barely improving on the trivial factor $n$ approximation obtained by broadcasting each client's message individually. 
In \cite{chlamtavc2014linear}, for a graph with minrank $k$, a scalar linear index coding scheme with an approximation factor of $n^{1-\epsilon_k}, \epsilon_k \to 0$ as $k \to \infty$, was provided which is nontrivial for a constant $k$. 
In the negative direction, it has been shown that finding any constant-factor approximation of $\Ind(\G)$ in general is at least as hard as some well-known open problems in graph coloring \cite{DBLP:journals/tit/LangbergS11}. In this paper, we explore two different approaches to make progress despite this difficulty. The first approach is to restrict the side information structure to some specific type of graph, and attempt to exploit its properties to attain better approximations than what are possible in general. The second is to find ways of strictly improving the existing schemes for the general case, though we cannot quantify the improvement asymptotically.

For perfect graphs (a class including all bipartite graphs which will be defined in \cref{sec:defs}), it has been known for some time that the index coding rate can be computed exactly, as it is sandwiched between two graph parameters that are equal \cite{bar2011index}. For more general classes than this, exactly computing the broadcast rate seems too much to ask, and we seek instead to approximate it as best as possible. There has been some work already in the area of approximating $\Ind(\G)$ for restricted graph classes: in \cite{Arbabjolfaei016}, Arbabjolfaei and Kim show a simple 4-approximation of $\Ind(\G)$ (meaning the returned solution has rate at most $4 \cdot \Ind(\G)$) for undirected planar graphs; in \cite{DBLP:journals/corr/Mazumdar0V17} Mazumdar et al.  improve this to obtain a 2-approximation of $\Ind(\G)$ for undirected planar graphs. In the (even more restricted) outerplanar case, while the scalar linear index coding rate with a fixed-size alphabet is studied in \cite{DBLP:conf/isit/BerlinerL11} (it is in fact shown to be equal to the size of the minimum clique cover of $\G$), the nonlinear rate has not been studied beyond the known results for planar graphs. In general, it has been shown that the linear and nonlinear index coding rates can be extremely far apart, so the nonlinear case merits study even when the linear case is solved \cite{DBLP:journals/tit/LubetzkyS09,blasiak2011lexicographic}. The main technique used to approximate $\Ind(\G)$ for planar graphs is to exploit the ``dual'' relationship between $\Ind(\G)$ and another, easier to approximate quantity called the storage capacity, or $\Cap(\G)$, which was introduced in \cite{DBLP:journals/tit/Mazumdar15}. The relationship between these quantities is also used in \cite{DBLP:journals/corr/Mazumdar0V17} to show some lower bounds on $\Ind(\G)$ for very restricted graph classes such as odd cycles. We will make use of this general technique as well, and will define $\Cap(\G)$ and explore its relationship with $\Ind(\G)$ further in \cref{sec:defs}.

In the general  case (recall this includes directed graphs), there have been a series of works providing increasingly better schemes. Birk and Kol \cite{birk1998informed} provided the first such scheme when introducing the problem, the ``clique cover'' scheme, in which the side information graph is covered by as few vertex-disjoint cliques as possible. In this scheme the broadcaster transmits a single message for each clique, which is the sum (as vectors with entries in $\mathbb{F}_q$) of the vectors desired by each node in the clique. Such a clique covering is equivalent to a proper coloring of the complementary graph. This idea was further extended in \cite{shanmugam2014} to show that in fact a weaker notion of coloring called a ``local coloring'' of the complementary graph yields an index coding scheme as well. Another generalization of the clique cover scheme that was known as early as \cite{birk1998informed} is to instead cover by ``partial cliques,'' which are nearly-complete subgraphs.

More recently in \cite{arbabjolfaei2014local}, ideas from both the local coloring and partial clique cover schemes were merged into a linear program (LP)-based scheme which outperforms both schemes individually. We continue in this line of work, showing a novel LP-based index coding scheme which combines ideas from previous schemes in order to obtain strictly better performance. Our scheme can also be extended to generalize the scheme proposed in \cite{thapa2015generalized}, which proposed to cover the side information graph by a type of generalized cycle, rather than by cliques or partial cliques.

\subsection{Contributions}

All our contributions consist of (vector linear) index coding schemes, in various settings, as opposed to lower bounds on $\Ind(\G)$. Additionally, all our schemes correspond to solutions of particular linear programs, which will be described in more detail in \cref{sec:defs,sec:results}. For special graph families, we have chosen to focus specifically on undirected graphs, both for the sake of simplicity and for parity, as one family we consider (disk graphs) has no directed analogue. In the general case, we consider directed graphs as well.

\subsubsection{Approximations for Special Graph Families}
Continuing the line of work in \cite{DBLP:journals/corr/Mazumdar0V17}, we generalize beyond the case of undirected planar graphs to any undirected graph with small chromatic number. 
We prove new bounds on $\Cap(\G)$ and $\Ind(\G)$ that recover the results of \cite{DBLP:journals/corr/Mazumdar0V17} for planar graphs, give superior results for 3-colorable graphs, and also give constant-factor approximations for graphs with constant chromatic number $> 4$. The techniques used for these types of graph and the barriers to progress that seem to arise give insight about other cases as well; as evidence of this, we use some of the same bounds used to prove results about $k$-colorable graphs in order to improve the best known approximation of $\Ind(\G)$ for undirected sparse graphs with $o(n^2)$ edges.

The other main graph class we consider is more practically motivated. If our graph arises from thresholding the latencies between pairs of servers to 0 or 1, and these latencies roughly correspond to physical distances between servers in the real world, then we should expect two servers that are physically close to have an edge between them, and two servers that are far apart to not have an edge between them. This is very close to the notion of a ``unit disk graph,'' which is a graph formed by placing points in the plane that correspond to the vertices, and having an edge between two vertices whenever the corresponding points are less than some distance apart (we define this more formally in the next section). These graphs are thought to be good approximations of certain kinds of real-world networks, and in particular have seen widespread use in the area of scheduling problems for broadcast networks \cite{hale1980frequency, huson1995broadcast}. In this setting there are many broadcasters which each have some radius in which they broadcast, and we may wish to, for instance, assign frequencies to each broadcaster so that no two broadcasters in the same area are broadcasting on the same frequency. This can be viewed as a coloring problem on a disk graph, where colors correspond to frequencies, and broadcasters correspond to vertices.

There are also prior examples of hard problems which are very difficult to approximate for general graphs, but for which good approximations exist when restricting to unit disk graphs; for example, it is well-known that maximum independent set cannot be approximated within any constant factor (in polynomial time) in general, but when restricting to unit disk graphs there is a polynomial time approximation scheme \cite{DBLP:journals/jal/HuntMRRRS98}. We show improvements over the general approximation of $\Cap(\G)$ for a superclass of unit disk graphs, as well as constant-factor but potentially inefficient approximations of $\Ind(\G)$ for unit disk graphs, which can be made efficient in some special cases.

\subsubsection{Improved Schemes for the General Case}

One of the earliest index coding schemes for the general case is the simple ``clique cover'' scheme, and since its introduction various different generalizations have been provided, such as ``local graph coloring'' and the ``partial clique cover'' scheme. Our work, expanding on a previously introduced interference alignment approach, gives a method that combines many of these ``orthogonal'' generalizations together. We give an example of a side information graph which shows that our new method can provide strict improvement over previous approaches. Furthermore, using ideas and tools from the previous scheme we further generalize another scheme which exploits what are called ``Generalized Interlinked Cycles'' in the side information graph.

\subsubsection{Paper Overview}
The remainder of the paper is organized as follows:
\begin{itemize}
\item In \cref{sec:defs} we introduce some definitions and notation that is needed to state and prove our main results.
\item In \cref{sec:results} we summarize our main results, including several constant-factor approximations for special graph families, and improved schemes for the general case. Proofs are postponed until the next section.
\item In \cref{sec:proofs} we state and prove bounds from which the quality of our approximations follows for special graph families, and prove the correctness of the schemes for the general case. For the special graph families, many of the bounds proved here actually imply good approximations for more general classes of graph than those focused on in the previous section, but we have chosen to highlight the results for those specific types of graph for greater clarity of exposition.
\item In \cref{sec:IC_constructions} we provide detailed constructions of the improved schemes for general graphs presented in \cref{sec:results}.
\item In \cref{sec:further} we explore some difficulties in improving certain results further, including examples that demonstrate barriers to the success of some current proof techniques. We also discuss several interesting open questions and potential improvements to our results.
\end{itemize}

\section{Prerequisites}
\label{sec:defs}

Let us define our notation for sets, graphs, vectors, and matrices at the outset.
\begin{itemize}
	\item For any $n \in \integers^+$, $[n] \triangleq \set{1,2,\ldots,n}$.
	\item For any $n \in \integers^+$, $[m,n] \triangleq \set{m,m+1,\ldots,n}, m\leq n$.
	\item The complement of a set $A$ is denoted by $\overline{A}$.
	\item For a graph $\cG$, $\overline{\cG}$ denotes the directed complement of $\cG$.
	\item  For any set $A = \set{i_1, i_2, \ldots, i_r} \subseteq [n]$ and set of vectors $\set{\bf{v}_i}_{i\in [n]}$, $\bf{v}_{A}$ denotes the set $\set{\bf{v}_{j}}_{j\in A}$ and $\bf{v}_{[A]}$ denotes the matrix $[\bf{v}_{i_1}\; \bf{v}_{i_2}\; \ldots \bf{v}_{i_r}]$. For a matrix $G \in \F^{k\times n}$, $G_{[B]}$ denotes the sub-matrix of $G$ constructed from the columns of $G$ corresponding to $B \subseteq [n]$.
	\item   For a graph $\G$, $N(v,\cG) \subseteq V(\cG)\setminus \set{v}$ denotes the set of out-neighbors of $v\in V$. When the graph is clear from context, we shorten this to $N(v)$.
	\item An $[n,k]$-MDS matrix is a matrix in $\F^{k\times n}$, $\F$ any field, with the property that any $k$ column vectors of the matrix are linearly independent.
\end{itemize}

Given a graph $\G$ and a subset of vertices $V' \subseteq V(\G)$, we write $\G \rvert_{V'}$ to mean the subgraph of $\G$ induced on $V'$. We write $\alpha(\G)$ for the size of the maximum independent set of $\G$, i.e., the size of the largest set $V'$ such that $\G \rvert_{V'}$ is edgeless. Many of our results give approximations with quality depending on the \textbf{chromatic number} $\chi(\G)$, the minimum number of colors needed to color the vertices of $\G$ such that no two adjacent vertices have the same color (such a coloring is called a ``proper coloring''). Some results also make use of a related quantity, called the \textbf{local chromatic number} $\chi_l(\G)$, which is the maximum number of colors in any out-neighborhood $N(v,\cG)\cup v$ of a vertex $v \in V(\cG)$, minimized over all proper colorings of $\cG$. A few results depend also on the size of the largest clique (complete subgraph) in $\G$, the \textbf{clique number} written $\omega(\G)$.

A \textbf{planar graph} is a graph with an embedding into the plane such that no two edges cross. An \textbf{outerplanar graph} is a planar graph, with the additional restriction that it has an embedding into the plane such that all vertices lie on the exterior face of the graph (i.e. a drawing exists with no vertex enclosed by edges). A \textbf{perfect graph} is a graph with the property that for every induced subgraph $\G'$, $\omega(\G') = \chi(\G')$. This class includes all bipartite graphs, and it is also known that the complement of every perfect graph is perfect.

Another type of graph we consider here are ``disk graphs,'' often thought to be good models of real-world networks where connections between nodes are based on their proximity in some metric. Disk graphs are a special case of \textbf{geometric intersection graphs}; these are the graphs which can be formed by placing shapes (usually of some restricted form) in the plane (or sometimes a higher dimensional space), then associating each shape with a vertex, and defining two vertices to have an edge whenever their corresponding shapes overlap (or touch at a single point). Any layout of shapes in the plane which corresponds to a specific graph $\G$ in this way is called a \textbf{geometric representation} of $\G$. Whenever a graph has such a geometric representation, we say it is an \textbf{intersection graph}. In a \textbf{disk graph}, we require that the graph has a geometric representation where all shapes are circles, but of possibly varying sizes. In a \textbf{unit disk graph}, or UDG, we further require that all such circles have unit radius, i.e. radius 1. We will even consider a special case of unit disk graphs, introduced in \cite{DBLP:journals/tcs/DiazK07}, called \textbf{$\lambda$-precision unit disk graphs}, which are those unit disk graphs for which there exists a geometric representation where every pair of disk centers is distance at least $\lambda$ from one another.
 
We say a subset $V' \subseteq V(\G)$ of vertices is a vertex cover of $\G$ if every edge of the graph includes some vertex in $V'$. We denote by $\VC(\G)$ the minimum size of all such covers. We can relax the notion of a vertex cover to the following LP, of which we refer to the solution as the \textbf{minimum fractional vertex cover}, with value $\FVC(\G)$:
 \begin{align*}
\textrm{min.} \hspace{1cm} & \sum_{v \in V(\G)} x_v &\\
\textrm{s.t.} \hspace{1cm} & x_u + x_v \geq 1 & \textrm{for every edge $(u,v) \in E(\G)$}  \\
& 0 \leq x_v \leq 1 & \forall v \in V(\G).
 \end{align*}
 
 A matching $E' \subseteq E(\G)$ in a graph $\G$ is a subset of edges with the property that no vertex of $\G$ is adjacent to more than one edge of $E'$. We write $\MM(\G)$ for the size of the maximum matching of $\G$. Similar to vertex cover, we can relax this notion to the following LP for \textbf{fractional maximum matching}, the optimal value of which we denote by $\FMM(\G)$:
\begin{align*}
\textrm{max.} \hspace{1cm} & \sum_{e \in E(\G)} y_e &\\
\textrm{s.t.} \hspace{1cm} & \sum_{e \in E(\G) : v \in e} y_e \leq 1 & \forall v \in V(\G)  \\
& 0 \leq y_e \leq 1 & \forall e \in E(\G).
 \end{align*}
It is well-known that this is the dual LP to that for fractional vertex cover, and thus by duality we have for any graph $\FMM(\G) = \FVC(\G)$.
 
As mentioned briefly in \cref{sec:intro}, it will be useful for us to consider another graph parameter which turns out to be closely related to the index coding rate, called the storage capacity of the graph, or $\Cap(\G)$. Intuitively, the storage capacity corresponds to the maximum size of an error-correcting code in which each vertex of the graph stores a symbol from $\mathbb{F}_q$, and we require that if any single vertex fails (in a detectable way) and its data becomes inaccessible, the $q$-ary symbol stored at it can be recovered as a function of only that vertex's neighbors in the graph. Thus if the graph is complete, this reduces to the notion of a single-erasure correcting code, as then there are no restrictions on which locations can be accessed to recover.

Formally, we say a set of codewords $\mathcal{C} \subseteq \Sigma^n$ is a recoverable distributed storage system code for the graph $\G$ with $|V(\G)| = n$ over alphabet $\Sigma, |\Sigma| = q,$ if there exist decoding functions $g_i: \Sigma^{|N(v_i)|} \rightarrow \Sigma, i =1,2, \dotsc, n$ such that for any codeword $X = (X_1, X_2, \dotsc, X_n) \in \mathcal{C}$, $g_i(\set{X_j : j \in N(v_i)}) = X_i$ for all $i$. We are primarily interested in the question of how large any such code can be over some particular network; to this end we define the  storage capacity
\begin{equation}
\Cap_q(\G) = \max_\mathcal{C} \log_{q} |\mathcal{C}|
\end{equation}
where the maximum is taken over all recoverable distributed storage system codes over an alphabet of size $q$, and 
we then define the overall capacity to be
\begin{equation}
\Cap(\G) = \sup_{q \ge 2} \Cap_q(\G).
\end{equation}

One of the main results of \cite{DBLP:journals/tit/Mazumdar15} proves the following somewhat unexpected dual relationship between the storage capacity and the index coding rate for $\G$ with $|V(\G)| = n$:
\begin{equation}
\Cap(\G) = n - \Ind(\G).
\end{equation}
Thus finding either quantity exactly is equivalently hard, though there is no reason to expect the two to be equally hard to approximate, and indeed it seems generally to be the case that $\Ind(\G)$ is much harder to approximate than $\Cap(\G)$. We will see later on that we are sometimes able to exploit the relationship between these two quantities to give guarantees about the quality of certain approximations -- in particular leveraging bounds on $\Cap(\G)$ to get at the otherwise difficult to approximate $\Ind(\G)$.

It is also shown in \cite{DBLP:journals/tit/Mazumdar15} that $\Cap(\G)$ is sandwiched between the size of the maximum matching of $\G$ and the minimum vertex cover of $\G$, which is used in proving the results for planar graphs. The fact that taking one vertex from each edge in a maximum matching yields a feasible vertex cover implies these two quantities are at most factor 2 apart, so this yields a simple 2-approximation of $\Cap(\G)$ for any graph. Thus when we try to approximate $\Cap(\G)$ for restricted $\G$, we are primarily interested in improving on the 2-approximation, whereas for $\Ind(\G)$, almost any nontrivial approximation is of interest.

The primary quantity we will use to approximate the storage capacity of a graph is the \textbf{maximum fractional clique packing} of $\G$, an LP relaxation of clique packing in which we try to pack as many large cliques within $\G$ as possible. Specifically, we write $\FCP(\G)$ for the solution to the following LP, where $\mathcal{K}$ denotes the set of all cliques in $\G$:
\begin{align*}
\textrm{max.} \hspace{1cm} & \sum_{C \in \mathcal{K}} x_C(|C| - 1) &\\
\textrm{s.t.} \hspace{1cm} & \sum_{C \in \mathcal{K} : v \in C} x_C \leq 1 & \forall v \in V(\G) \\
& 0 \leq x_C \leq 1 & \forall C \in \mathcal{K}.
 \end{align*}
 Note that in general we may not be able to compute the solution to this LP efficiently without a bound on the size of the largest clique in $\G$. The main reason $\FCP(\G)$ proves useful as an approximation of $\Cap(\G)$ is due to the bound
 \begin{equation*}
 \FCP(\G) \leq \Cap(\G)
 \end{equation*}
 shown in \cite{DBLP:journals/corr/Mazumdar0V17}. For approximating the index coding rate of $\G$ rather than the capacity, we will use the complementary quantity $\FCC(\G)$, the size of the \textbf{minimum fractional clique cover} of $\G$, where we instead seek to use as few cliques as possible in order to cover every vertex of $\G$ by some clique. This quantity is equal to the solution of the following LP:
\begin{align*}
\textrm{min.} \hspace{1cm} & \sum_{C \in \mathcal{K}} y_C &\\
\textrm{s.t.} \hspace{1cm} & \sum_{C \in \mathcal{K} : v \in C} y_C \geq 1 & \forall v \in V(\G) \\
& 0 \leq y_C \leq 1 & \forall C \in \mathcal{K}.
 \end{align*}
 It is a simple exercise to see that $\FCC(\G) = n - \FCP(\G)$, so we will sometimes use these two notations interchangeably depending on what is most convenient. The above relationship between $\FCP(\G)$ and $\Cap(\G)$ also immediately yields the upper bound
 \begin{equation}
 \Ind(\G) \leq \FCC(\G),
 \end{equation}
which has been known for some time in the index coding literature \cite{DBLP:journals/corr/abs-1004-1379}. 
 
 Another bound on $\Ind(\G)$ which we will rely on heavily in our approximations, first shown in \cite{bar2011index}, is that $\Ind(\G)$ is lower bounded by the size of the \textbf{maximum acyclic induced subgraph} of $\G$, or $\MAIS(\G)$. For undirected $\G$, $\MAIS(\G) = \alpha(\G)$, but in general for directed $\G$ we have only $\alpha(\G) \leq \MAIS(\G)$, as every independent set clearly induces an acyclic subgraph. So it is always true that
 \begin{equation}
 \alpha(\G) \leq \MAIS(\G) \leq \Ind(\G) \leq \FCC(\G).
 \end{equation}
 From this we can see why it is easy to find $\Ind(\G)$ exactly if $\G$ is perfect, as then $\overline{\G}$ is perfect also, so if we write $\CC(\G)$ for the minimum integral clique cover of $\G$, we have
 \begin{equation}
 \omega(\overline{\G}) = \alpha(\G) \leq \Ind(\G) \leq \FCC(\G) \leq \CC(\G) = \chi(\overline{\G}),
 \end{equation}
 and the leftmost and rightmost terms are equal as $\overline{\G}$ is perfect. While both $\omega(G)$ and $\chi(G)$ are NP-hard to compute in general, we can instead compute any more nicely-behaved quantity sandwiched between them, such
 as the Lov\'asz theta function $\theta(\G)$.
 
Finally, we will in certain cases wish to cover the graph instead by a generalization of a clique, called a \textbf{$k$-partial clique}. A $k$-partial clique on $n$ vertices is a subgraph in which every vertex has at least $n - k - 1$ out-neighbors, and at least one vertex has exactly $n - k - 1$ out-neighbors. Thus, a complete subgraph on $n$ vertices is a $0$-partial clique.

\section{Main Results}
\label{sec:results}

In this paper, we present primarily two types of results for approximating the index coding rate of a graph: those which apply only to graphs in specific families, and those which apply to general graphs. When working with a special family, we can often provide good approximations of the index coding rate by using simple schemes but leveraging properties of the graph to prove these simple schemes are effective. In contrast, as it is known to be difficult to approximate the index coding rate in the general case, most of our results in the general (directed) setting do not provide provably good approximations; instead, they can be viewed as methods of strengthening the simple schemes to ones that perform strictly better, although we are not always able to rigorously quantify how much better they perform.

\subsection{Approximation Results for Special Graph Families}
\label{sec:family_results}

Most of the results in this paper relating to specific graph families do not depend fundamentally on the graph family itself, but rather on certain nice properties of the graph family such as small chromatic number. In this section we do not state our results in full generality or prove them, but instead give instantiations of the general results with respect to the graph families we are most interested in. The most general versions of these results are stated and proven in \cref{sec:proofs}.

At a high level, the common technique used in these results 
is to employ the (relatively) easy-to-compute quantity $\FCP(\G)$ as an approximation of $\Cap(\G)$, and similarly to use $n - \FCP(\G) = \FCC(\G)$ as an approximation of $\Ind(\G)$. The main challenge comes in proving the quality of these approximations. The table below summarizes the state-of-the-art bounds for the main graph families considered in this paper. We reiterate that in this subsection, all results assume the graph is undirected.

\renewcommand{\arraystretch}{1.5}
\begin{center}
\begin{threeparttable}[htbp]
\caption{Best-Known Approximations of $\Ind(\G)$ and $\Cap(\G)$}
\label{results_table}
\begin{centering}
\begin{tabular}{|m{30mm}|c|c|}
\hline
 \centering \textbf{Graph Type} & \textbf{UB for $\Cap(\G)/\FCP(\G)$} &  \textbf{UB for $\FCC(\G)/\Ind(\G)$}\\
 \hhline{|=|=|=|}
\centering Unrestricted & $2$ & $O(n \frac{\log \log n}{\log n})$  \\
 \hline
 \centering Sparse Graph ($|E(\G)| = O(n^{1+\epsilon}))$ & $2$ & $O(n^\epsilon)^{**}$  \\
 \hline
\centering Small Chromatic Number ($\chi(\G) = k \geq 2$) & \ $2 - \frac{2}{k}^*$ & \ $\frac{k}{2}^*$  \\
\hline
\centering General Disk Graph & \ $\frac{3}{2}^{*}$ & $O(n \frac{\log \log n}{\log n})$  \\
\hline
\centering Unit Disk Graph & \ $\frac{3}{2}^{*}$ & $3^*$  \\
\hline
\centering $\lambda$-precision UDG, $\lambda \leq 1/\sqrt{2}$ & \ $\frac{3}{2}^*$ & $\frac{64}{\lambda^2} + 1^*$  \\
\hline
\end{tabular}
\begin{tablenotes}
\footnotesize
\item $^*$Bound proved in this work.
\item $^{**}$Bound in this work improves previous best bound by a constant factor.
\end{tablenotes}
\end{centering}
\end{threeparttable}
\end{center}

\subsubsection{Results for Graphs with Small Chromatic Number}
Many of the results in \cite{DBLP:journals/corr/Mazumdar0V17} are aimed at approximating $\Cap(\G)$ and $\Ind(\G)$ in the case that $\G$ is planar, often by exploiting the 4-colorability of planar graphs. Here we generalize these ideas further to the case that $\G$ is $k$-colorable for some $k$. Our first result generalizes the $\frac{3}{2}$-approximation of $\Cap(\G)$ for planar $\G$ in \cite{DBLP:journals/corr/Mazumdar0V17} to a $(2 - \frac{2}{k})$-approximation when $\G$ is $k$-colorable.

\begin{theorem}
\label{chromatic_cap}
If $\G$ has $\chi(\G) = k \geq 2$, then
\begin{equation}
\frac{\Cap(\G)}{\FCP(\G)} \leq 2 - \frac{2}{k}.
\end{equation}
\end{theorem}

Similarly, \cite{DBLP:journals/corr/Mazumdar0V17} presents a 2-approximation of index coding rate for planar graphs. By generalizing their bound to exploit $k$-colorability instead of 4-colorability we immediately obtain an approximation for $k$-colorable graphs, but the quality of this bound scales poorly with $k$. However, we can use a different technique to show $\FCC(\G)$ is a $\frac{k}{2}$-approximation for $k$-colorable $\G$.

\begin{theorem}
\label{chromatic_ind}
If $\G$ has $\chi(\G) = k \geq 2$, then
\begin{equation}
\frac{\FCC(\G)}{\Ind(\G)} \leq \frac{k}{2}.
\end{equation}
\end{theorem}

\subsubsection{Results for Sparse Graphs}
Many of our results, especially for approximating $\Ind(\G)$, rely on the fact that graph families with small chromatic number always contain a relatively large independent set. This fact combined with the chain of inequalities $\alpha(\G) \leq \Ind(\G) \leq \FCC(\G)$ and bounds on $\FCC(\G)$ is often enough to give good results in the special cases we consider. The following theorem attempts to generalize this idea as much as possible, by using Tur\'an's theorem to guarantee the existence of a large independent set in any sufficiently sparse graph. If we restrict back to the planar or outerplanar case, this result is weaker than the other more specialized results.
\begin{theorem}
\label{sparse_ind}
Let $\G$ be a graph with $n$ vertices and $e$ edges. Then
\begin{equation}
\frac{\FCC(\G)}{\Ind(\G)} \leq \max\left(\frac{e(n-2)}{n(n-1)} + 1, \frac{2e}{3n} + \frac{4}{3}\right).
\end{equation}
\end{theorem}

\subsubsection{Results for Disk Graphs}

As mentioned previously, the other main graph family we will consider are the disk graphs, and in particular unit disk graphs. The primary difficulty with this graph family which does not occur in the case of planar or outerplanar graphs is that these graphs may be very dense and contain cliques of arbitrarily large size, which means that in general they do not have linear-sized independent sets. If $\alpha(\G)$ is very small, then the lower bound $\alpha(\G) \leq \Ind(\G)$ becomes very weak, and approximating $\Ind(\G)$ becomes difficult. The situation is better for approximating the storage capacity, since the corresponding inequality is $\Cap(\G) \leq n - \alpha(\G)$, meaning when $\alpha(\G)$ is very small $\Cap(\G)$ is easy to approximate. We use this idea along with some facts about disk graphs to get the following approximation guarantee.

\begin{theorem}
\label{disk_cap}
If $\G$ is a disk graph, then
\begin{equation}
\frac{\Cap(\G)}{\FCP(\G)} \leq \frac{3}{2}.
\end{equation}
\end{theorem}

When $\G$ is a disk graph or even a unit disk graph, it becomes increasingly difficult to approximate $\Ind(\G)$ using preexisting methods as $\G$ contains larger and larger cliques. If we are willing to tolerate superpolynomial running time (which may be reasonable, as finding $\Ind(\G)$ exactly is not even known to be in NP), we can use a result of \cite{DBLP:journals/endm/ChalermsookV16} along with some results from the disk graph literature to obtain the following approximation.

\begin{theorem}
\label{udg_ind_bound}
If $\G$ is a unit disk graph, then
\begin{equation}
\frac{\FCC(\G)}{\Ind(\G)} \leq 3.
\end{equation}
\end{theorem}

If instead we insist on polynomial running time, we cannot prove a constant-factor approximation for all UDGs (the LP which has $\FCC(\G)$ as its solution may have a superpolynomial number of constraints), but we can recover good approximations in some special cases.

\begin{theorem}
\label{clique_udg_ind}
If $\G$ is a unit disk graph with clique number $\omega(\G)$, then
\begin{equation}
\frac{\FCC(\G)}{\Ind(\G)} \leq \omega(\G) + 1,
\end{equation}
and furthermore we can obtain an approximation of $\Ind(\G)$ with this approximation factor in polynomial time.
\end{theorem}

In \cite{DBLP:journals/jal/HuntMRRRS98}, Hunt et al. introduced the notion of ``$\lambda$-precision unit disk graphs.'' These are unit disk graphs with the additional constraint that the centers of every pair of disks are at distance at least $\lambda$ from each other, which may be a reasonable constraint in some real-world scenarios. This allows us to prove a bound on the clique number in terms of $\lambda$, which we can translate into a bound on $\Ind(\G)$ using \cref{clique_udg_ind}.

\begin{theorem}
\label{lambda_udg_ind}
If $\G$ is a $\lambda$-precision unit disk graph, then
\begin{equation}
\frac{\FCC(\G)}{\Ind(\G)} \leq \frac{64}{\lambda^2} + 1 = O(\lambda^{-2}) + 1,
\end{equation}
and furthermore we can obtain an approximation of $\Ind(\G)$ with this approximation factor in polynomial time.
\end{theorem}

\subsection{Algorithms for General Graphs}
\label{sec:algo}

As seen above, almost all our results approximating the index coding rate of graphs from special families use the fractional clique cover as the achievability scheme. In this section we instead describe more complex vector linear achievability schemes which strictly improve upon the fractional clique cover, and thus can be viewed as a further strengthening of the approximations described previously for special graph families. Although we know of specific examples where these new schemes are superior, we leave as an open question whether they can yield better constant-factor approximations for certain graph families than those attained by $\FCC(\cG)$. In this subsection we consider directed as well as undirected graphs. The detailed proofs of the results in this subsection are postponed to \cref{sec:IC_constructions}. 

Let us first look at the index coding problem from an interference alignment perspective. Suppose that the data requested by user $i$ (vertex $v_i$) is $\bfx_i \in \F^\ell$. We assign a vector $\bf{v}_i$ to each vertex $v_i \in V(\cG)$ such that the vectors satisfy the following condition,
\begin{equation}\label{interference_align_condition}
	\bf{v}_i \not\in \spn(\bf{v}_{N(v_i,\overline\cG)}).
\end{equation}
From the interference alignment perspective, $N(v,\overline\cG)$ are the interfering set of indices for user $v$. Recall we define $\bfv_{[V(\cG)]} \triangleq [\bfv_1\; \bfv_2\;\ldots\;\bfv_n]$. The index code (broadcaster transmission) is given by $ \bfv_{[V(\cG)]}\cdot [\bf{x}_1\; \bf{x}_2\; \ldots \bf{x}_n]_{\ell \times n}^T \in \F_q^{\ell \times \rank{\bfv_{[V(\cG)]}}}$. It can be seen that each node $v_i$ can recover $\bf{x}_i$ from the index code because of \cref{interference_align_condition}.

In this section, we utilize the interference alignment perspective to find algorithms that improve beyond $\FCC(\cG)$. We begin by combining two orthogonal generalizations of $\FCC(\cG)$.

\subsubsection{Local Chromatic Number and Partial Clique Cover}
\label{sec:LPCC}

It is certainly possible to satisfy the requirements in \cref{interference_align_condition} if $\dim(\spn(\bf{v}_{[n]})) = n$, however, our goal is to minimize the dimension of $\spn(\bf{v}_{[n]})$. One solution to this problem is to find a  {{proper coloring}} of the graph $\overline{\cG}$ and assign orthonormal vectors to each color class (the same vector is assigned to all vertices with the same color). Thus, an achievable broadcast rate is given by the chromatic number of ${\overline{\cG}}$. Note that the size of a minimum (integral) clique cover of a graph $\G$ is the same as the chromatic number of the complementary graph $\overline{\G}$, and similarly $\FCC(\G) = \chi_F(\overline{\cG})$, the fractional chromatic number of $\overline\cG$.

One way to improve beyond the fractional clique cover scheme is the local chromatic number. The local chromatic number of $\overline{\cG}$ is always less than (or equal to) $\chi(\overline \G)$. Using the interference alignment perspective it is easy to see that we can assign the column vectors from an $[n,\chi_\ell(\overline \cG)]$-MDS matrix to attain an index coding rate equal to the local chromatic number as shown in \cite{shanmugam2013local}. A linear relaxation of the integer program corresponding to the local chromatic number gives a vector linear index coding scheme better than $\chi_F(\overline\cG) = \FCC(\G)$.

Another approach to improving the clique cover is to instead find a {\em partial clique cover} of $\cG$ \cite{birk1998informed}. Whereas a clique cover is a cover of the vertices of the graph by complete subgraphs, a $k$-partial clique cover is instead a cover of the vertices of the graph by $k$-partial cliques, which were defined in \cref{sec:defs}. Let $k_\cS$ be the smallest $k$ such that $\cS\subset V$ is a $k$-partial clique. In each of the $k_\cS$-partial cliques $\cS$, one can use a $[\abs{\cS},k_{\cS}]$-MDS matrix to assign vectors to the nodes to satisfy \cref{interference_align_condition}.

We can in fact go further, and combine the partial clique cover and the local chromatic number schemes to obtain an index code which generalizes both these schemes, as shown in \cref{thm:main}. In some cases \cref{local_partial_IP} provides strictly better solutions than either the partial clique cover or the local chromatic number of $\cG$.

\begin{theorem}\label{thm:main}
The minimum broadcast rate of an index coding problem on the side information graph $\G$ is upper bounded by the optimum value of the following linear program, where $\cK \triangleq 2^{V(\G)}.$
\begin{subequations}\label{local_partial_IP}
\begin{flalign}
\min\;\;\;  & \;\;\;t \nonumber\\
\rm{s.t.}\;\;\;\;&  \sum_{\cS \in \cK} \min\set{\abs{\cS\cap N(v,\overline \cG)},k_\cS+1}\rho_\cS \le t, \;\;v \in V(\cG) \label{local_partial_IP_minmax_constr}\\ 
	&  \displaystyle\sum_{\cS \in \cK : v \in \cS}  \rho_\cS \geq 1, \;\; v \in V(\cG)  \label{local_partial_IP_cover_constr}\\
	& \rho_\cS \in [0,1], \;\;\cS \in \cK. \label{local_partial_IP_frac_constr}
\end{flalign}
\end{subequations}
\end{theorem}


Let us explain the term $$\sum_{S \in \cK} \min\set{\abs{S\cap N(v,\overline{\cG})}, k_\cS+1}\rho_s$$ in \cref{local_partial_IP_minmax_constr}, for the integer version of the above linear program. Let $\cS_1, \cS_2, \ldots, \cS_\tau$ be the set of selected partial cliques. Then, for each vertex  $v \in V(\cG)$ compute the sum $\displaystyle \sum_{i=1}^\tau \min\set{\abs{\cS_i\cap N(v,\overline{\cG})}, k_{\cS_i}+1}$. Thus each selected partial clique only contributes $\min\set{\abs{\cS_i\cap N(v,\overline{\cG})}, k_{\cS_i}+1}$. Now, the number of broadcast bits corresponds to the maximum sum for any vertex $v$, i.e. 
$$t = \max_{v\in V(\cG)}\sum_{\cS \in \cK} \min\set{\abs{\cS\cap N(v,\overline \cG)},k_\cS+1}\rho_\cS = \max_{v\in V(\cG)} \sum_{i=1}^\tau \min\set{\abs{\cS_i\cap N(v,\overline{\cG})}, k_{\cS_i}+1}.$$

A solution to the integral version of the above linear program corresponds to a scalar linear index code. From the linear program in \cref{local_partial_IP}, we instead obtain a vector linear index code, the details of which are covered in \cref{sec:IC_constructions}.

There is one more way we can generalize the solution of the linear program in \cref{local_partial_IP}, which is to recursively apply the linear program to subgraphs. The recursive linear program is given in the following theorem.

\begin{theorem}[Recursive LP]\label{thm:recur}
 Let $IC_{FLP}(\cG)$ denote the value of an optimal solution to the  linear program below for graph $\cG$:
\begin{equation}\label{local_partial_rec_lin_prog}
\begin{array}{ll@{}}
\min  & t \\
\rm{s.t.}& \sum_{\cS \in \cK} \min\set{\abs{\cS\cap \overline{N(v,\cG)}},IC_{FLP}(\cG\rvert_\cS)}\rho_\cS \le t, \;\; v \in V(\cG)  \\ 
	& \displaystyle\sum_{\cS \in \cK : v \in \cS}  \rho_\cS \geq 1, \;\; v \in V(\cG)  \\
	& \rho_\cS \in [0,1], \;\;\cS \in \cK.
\end{array}
\end{equation}
where $IC_{FLP}(\mathcal{H})$ is defined to be $1$ for single vertex graphs $\mathcal{H}$. Then the minimum broadcast rate of an index coding problem on the side information graph $\cG$ is bounded from above by $IC_{FLP}(\cG)$.
\end{theorem}

The index code corresponding to the linear program in \cref{thm:recur} can be easily obtained from the index coding solution for \cref{thm:main} as shown in \cref{sec:proofs}.
Let us now give an explicit example of a graph where our index coding scheme is a strict improvement over the existing schemes. Of course, since our scheme is more general, it is clear that its performance must be at least as good for every graph $\cG$. 

Consider the index coding problem described by the graph in \cref{fig:local_example}. For this graph, the index code based on the fractional local chromatic number has broadcast rate $4$,  the index code based on just the fractional partial clique clique cover has broadcast rate $11/3$ and the proposed scheme combining the local chromatic number and partial clique cover in \cref{local_partial_IP} has  broadcast rate $7/2$. Similarly, \cref{fig:comp_yhk} shows an example for which the recursive version of the proposed scheme in \cref{thm:recur} is a strict improvement over the corresponding recursive scheme proposed in \cite[theorem 4]{arbabjolfaei2014local}, with broadcast rates $3$ and $7/2$, respectively.

\begin{figure}
  \centering
 \includegraphics[width=6cm]{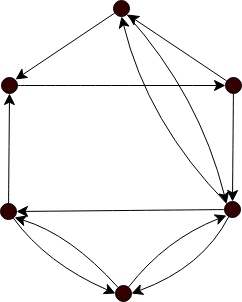}
\caption{Side information graph for which the broadcast rate of the proposed scheme in \cref{thm:main} is a strict improvement over the existing schemes (fractional local chromatic number and fractional partial clique cover).}	
\label{fig:local_example}
\end{figure}

\begin{figure}
  \centering
 \includegraphics[width=8cm]{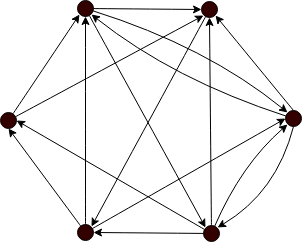}
\caption{Side information graph for which the broadcast rate of the proposed recursive scheme in \cref{thm:recur} is a strict improvement over the existing recursive schemes in \cite[theorem 4]{arbabjolfaei2014local}.} 
\label{fig:comp_yhk}
\end{figure}

\subsubsection{Generalized Interlinked Cycle Cover}
\label{sec:GIC}

We now generalize the fractional clique cover scheme in another direction. Since cycle and clique covers yield natural solutions to the index coding problem it makes sense to combine these structures to obtain a more general solution. The $n$-GIC (Generalized Interlinked Cycle) graph structure presented in \cite{thapa2015generalized} provides such a solution. Our contribution is to show that this scheme can be further generalized by combining it with the partial clique cover technique presented above. We will call the relevant graph structure used to cover the side-information graph a $(k,n_1)$-GIC; here we simply define this structure, and the details of the scheme will be postponed to \cref{sec:IC_constructions}.

We say a graph $\cG$ with $n$ vertices is a $(k,n_1)$-GIC if it has the following properties:
\begin{enumerate}
\item $\cG$ contains a set of $n_1$ vertices, denoted by $ V_{\mathrm{I}} $, such that for any vertex $v_i \in V_{\mathrm{I}}$ there are at least $n_1-k-1$ vertices $v_j\in V_{\mathrm{I}}$ with the property that there is a path from $ v_i $ to $ v_j $ which does not include any other vertex of $ V_{\mathrm{I}} $. We call $ V_{\mathrm{I}} $ the \emph{inner vertex set}, and let $ V_{\mathrm{I}}=\{v_1,v_2,\dotsc,v_{n_1}\} $. The vertices of $ V_{\mathrm{I}}$ are referred to as \emph{inner vertices}. 

\item Due to the above property, we can always find a \emph{directed rooted tree} (denoted by $T_i$) with maximum number of leaves in $V_\rmI$ and root vertex $ v_i $, having at least $n_1-k-1$  other vertices in $V_\rmI \setminus \{v_i\} $ as leaves. The trees may not be unique. Denote the union of all $n_1$ such trees by $D \triangleq  \bigcup_{i : v_i\in V_{\rmI}} T_i$. Then the digraph $D$ must satisfy the following two conditions:
\begin{property}\mbox{}\\[-\baselineskip]\label{Digraph_properties}
\begin{enumerate}
\item Every cycle in the digraph $D$ contains at least two vertices in the vertex set $V_\rmI$.
\item For all ordered pairs of inner vertices ($ v_i,v_j $), $ i\neq j $, there is only one path in $D$ from $v_i$ to $v_j$ that does not include any other vertices in $V_\rmI$.
\end{enumerate}
\end{property}
\end{enumerate}

\subsubsection{Example}

\begin{figure}
  \centering
 \includegraphics[width=6cm]{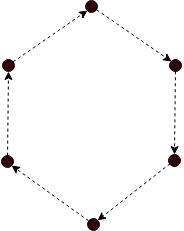}
\caption{Side information graph for which our proposed $(k, n_1)$-GIC scheme outperforms the $n$-GIC scheme of \cite{thapa2015generalized}. All edges are present except those indicated by dashed lines.}	
\label{fig:GICexample}
\end{figure}

We provide an example where the proposed GIC scheme performs strictly better than the GIC scheme in \cite{thapa2015generalized} in \cref{fig:GICexample}. The graph in \cref{fig:GICexample} has an index coding rate of $2$ using a partial clique cover scheme. Since the proposed GIC scheme is a generalization of partial clique covers it performs at least as well. 

A vector linear scheme using a fractional cover with the GIC scheme proposed in \cite{thapa2015generalized} gives an index coding rate of $5/2$. Note that for the graph proposed in \cref{fig:GICexample}, there is no GIC (as proposed in \cite{thapa2015new}) with inner vertex set of size $4$, since this violates condition $a)$ in \cref{Digraph_properties}.

\section{Proofs for Index Coding Rate Approximations}
\label{sec:proofs}

In this section we prove the results of \cref{sec:family_results}. Typically we will do so by establishing a more general result, from which we just need to plug in certain parameters of the graph family in question to obtain the more specific statement. To begin we consider bounds which exploit the graph having small chromatic number.

\subsection{Bounds Using Chromatic Number}
In \cite{DBLP:journals/corr/Mazumdar0V17}, several results showing constant-factor approximations for both storage capacity and index coding rate in planar graphs are given. For the most part, these results depend not specifically on the planarity, but on the small chromatic number of the graph in question, as well as the chromatic number of the subgraph induced by removing a maximal set of triangles. In particular, the techniques 
used to show a constant-factor approximation of $\Ind(\G)$ for planar graphs depend not only on the 4-colorability of planar graphs, but also on the 3-colorability of triangle-free  planar graphs. Here we generalize and extend these techniques to give approximations in terms of the chromatic number of the graph.

To begin, the same argument used in \cite{DBLP:journals/corr/Mazumdar0V17} to show clique packing is a $\frac{3}{2}$-approximation of $\Cap(\G)$ for planar graphs easily extends to show \cref{chromatic_cap_bound}; we reproduce essentially the same proof as that of \cite{DBLP:journals/corr/Mazumdar0V17} for completeness, as some of the intermediate steps will be useful in subsequent results. We will also make use of the fact, noted in \cite{DBLP:journals/tit/Mazumdar15}, that $\Cap(\G) \leq \VC(\G)$, the size of the minimum vertex cover.

\begin{theorem}
\label{chromatic_cap_bound}
Let $\G$ be a graph, $T$ be the vertices of a maximal set of $t = \frac{|T|}{3}$ vertex-disjoint triangles in $\G$, and $\G' = \G \rvert_{V(\G) \setminus T}$. Suppose the minimum vertex cover of $\G'$ has size $k$, and $\chi(\G') \leq l$. Then
\begin{equation}
\frac{\Cap(\G)}{\FCP(\G)} \leq \frac{3t + k}{2t + kl/(2l-2)}. 
\end{equation}
\end{theorem}
\begin{proof}
To start, we have the upper bound $\Cap(\G) \leq 3t+k$, assuming perfectly efficient storage on all triangles, and using the bound $\Cap(\G') \leq k$ on the remainder of the graph. We have also a lower bound $\FCP(\G) \geq 2t + \FCP(\G')$, by including each triangle in $T$ in the fractional clique packing, then using the optimal packing on $\G'$.

Then as $\G'$ is triangle-free, the maximum fractional clique packing is just a maximum fractional matching, which is equal to the minimum fractional vertex cover by duality. So to conclude, we need only bound the integrality gap of vertex cover on $\G'$. Suppose we have a fractional vertex cover with variables $x_{v_1}, \dotsc, x_{v_n}$. Vertex cover is $\frac{1}{2}$-integral, so assume all $x_{v_i} \in \set{0, \frac{1}{2}, 1}$, and as it is a fractional vertex cover, if $(v_i, v_j)$ is an edge, then $x_{v_i} + x_{v_j} \geq 1$. $\G'$ is $l$-colorable by assumption, so let $I_1, \dotsc, I_l$ be a partition of $\set{v_i : x_{v_i} = \frac{1}{2}}$ corresponding to an $l$-coloring of $\G'$, such that
\begin{equation*}
\sum_{v \in I_1} x_{v} \geq \sum_{v \in I_2} x_v \geq \cdots \geq \sum_{v \in I_l} x_v.
\end{equation*}
First note that if $l = 1$, there are no edges, so the integrality gap of vertex cover is 1. Otherwise, we construct an integral vertex cover $y_{v_1}, \dotsc, y_{v_n}$ as follows: if $x_{v_i}$ is integral, then $y_{v_i} = x_{v_i}$. Otherwise, if $x_{v_i} = \frac{1}{2}$ and $v_i \in I_1$, we set $y_{v_i} = 0$, and if $x_{v_i} = \frac{1}{2}$ but $v_i \not \in I_1$, we set $y_{v_i} = 1$. This is a vertex cover, because the only rounded-down variables were those $x_{v_i}$ with $v_i \in I_1$, and the other endpoint of any edge with $v_i$ must be in $I_2 \cup \cdots \cup I_l$, as the partition corresponds to a coloring. $I_1$ comprises at least a $\frac{1}{l}$-fraction of the rounded variables, so we rounded at most an $\frac{l-1}{l}$-fraction of variables up from $\frac{1}{2}$ to 1, thus 
\begin{equation*}
\sum_{v \in V(\G')} y_v \leq \frac{2(l-1)}{l} \sum_{v \in V(\G')} x_v.
\end{equation*}
This shows the integrality gap of vertex cover is at most $\frac{2l-2}{l}$, so
\begin{equation*}
\FCP(\G) \geq 2t + \FCP(\G') \geq 2t +  \frac{l}{2l-2} \cdot k.
\end{equation*}

Combining these two bounds, we have
\begin{equation*}
\frac{\Cap(\G)}{\FCP(\G)} \leq \frac{3t+k}{2t + (l/(2l-2)) \cdot k}.
\end{equation*}
\end{proof}

This bound itself will be useful for proving further bounds, but also immediately provides a guarantee on the approximation quality of $\FCP(\G)$ for graphs with small chromatic number, as if $\G'$ is a subgraph of $\G$, then $\chi(\G') \leq \chi(\G)$.

\begin{corollary}
Let $\G$ be a graph with $\chi(\G) = l$. Then
\begin{equation*}
\frac{\Cap(\G)}{\FCP(\G)} \leq \max\left(\frac{3}{2}, 2 - \frac{2}{l} \right).
\end{equation*}
\end{corollary}
\begin{proof}
If $l = 2, 3,$ or $4$, then $\frac{l}{2l-2} \geq \frac{2}{3}$, so
\begin{equation*}
\frac{3t+k}{2t + (l/(2l-2)) \cdot k} \leq \frac{3t+k}{2t + (2/3) \cdot k} = \frac{3}{2} \cdot \frac{t+k/3}{t+k/3} = \frac{3}{2}.
\end{equation*}
Otherwise $l \geq 5$, so $3t \leq \frac{4l-4}{l} \cdot t$. Then we have
\begin{equation*}
\frac{3t+k}{2t + (l/(2l-2)) \cdot k} \leq \frac{((4l-4)/l) \cdot t + k}{2t + (l/(2l-2)) \cdot k} = \frac{2l-2}{l} \cdot \frac{((4l-4)/l) \cdot t + k}{((4l-4)/l) \cdot t + k} = \frac{2l-2}{l} = 2 - \frac{2}{l},
\end{equation*}
so
\begin{equation*}
\frac{\Cap(\G)}{\FCP(\G)} \leq \max\left(\frac{3}{2}, 2 - \frac{2}{l} \right),
\end{equation*}
as desired.
\end{proof}

In the specific case that $\G$ is 3-colorable (such as when $\G$ is outerplanar), we can use this additional information along with an idea from the above proof to improve further.

\begin{theorem}
Let $\G$ be a graph with $\chi(\G) = 3$. Then 
\begin{equation*}
\frac{\Cap(\G)}{\FCP(\G)} \leq \frac{4}{3}.
\end{equation*}
\end{theorem}
\begin{proof}
Recall that fractional minimum vertex cover and fractional maximum matching are dual, so $\FMM(\G) = \FVC(\G)$ for all $\G$. We showed in the above proof that when $\chi(\G) = 3$, the integrality gap of vertex cover is at most $\frac{2 \cdot 3 - 2}{2} = \frac{4}{3}$, so we have $\frac{3}{4} \VC(\G) \leq \FVC(\G)$. As the maximum fractional matching is a feasible fractional clique packing with cliques of size at most $2$, we have $\FMM(\G) \leq \FCP(\G)$. In \cite{DBLP:journals/tit/Mazumdar15} it is observed that $\Cap(\G) \leq \VC(\G)$. Combining this, we have
\begin{equation*}
\frac{3}{4} \VC(\G) \leq \FVC(\G) = \FMM(\G) \leq \FCP(\G) \leq \Cap(\G) \leq \VC(\G),
\end{equation*}
thus $\FCP(\G)$ is within a $\frac{4}{3}$ factor of $\Cap(\G)$.
\end{proof}

\begin{corollary}
Let $\G$ be a graph with $\chi(\G) = k \geq 2$. Then
\begin{equation*}
\frac{\Cap(\G)}{\FCP(\G)} \leq 2 - \frac{2}{k}.
\end{equation*}
\end{corollary}

Now we move our attention to index coding. In the next two theorems, we 
provide two more general bounds on $\Ind(\G)$, each of which is a good approximation for certain special cases.

\begin{theorem}
\label{chromatic_ind_bound}
Let $\G$ be a graph with $\chi(\G) \leq j$, $T$ be the vertices of a maximal set of $t = \frac{|T|}{3}$ vertex-disjoint triangles in $\G$, $\G' = \G \rvert_{V(\G) \setminus T}$, and $k$ be the size of a minimum vertex cover of $\G'$. Suppose further that $\chi(\G') = l \geq 2$. Then
\begin{equation*}
\frac{\FCC(\G)}{\Ind(\G)} = \frac{n - \FCP(\G)}{\Ind(\G)} \leq j \cdot \frac{l-2}{2l-2} - j \cdot \frac{l-4}{2l-2} \cdot \frac{t}{n} + \frac{l}{2l-2}.
\end{equation*}
\end{theorem}
\begin{proof}
As seen in the proof of \cref{chromatic_cap_bound}, $\FCP(\G) \geq 2t +  \frac{l}{2l-2} \cdot k$ when $l \geq 2$. The size of the minimum vertex cover of $\G'$ is equal to the number of vertices of $\G'$ minus the size of the maximum independent set, so $k = n - 3t - \alpha(\G')$, thus
\begin{equation*}
n - \FCP(\G) \leq n - 2t - \left(\frac{l}{2l-2}\right) \cdot (n-3t -\alpha(\G')) = \frac{l-2}{2l-2} \cdot n - \frac{l-4}{2l-2} \cdot t + \frac{l}{2l-2} \cdot \alpha(\G').
\end{equation*}
For bounding $\Ind(\G)$, we have $\Ind(\G) \geq \alpha(\G) \geq \frac{n}{j}$. Then we simply combine the two bounds, using the fact that $\frac{\alpha(\G')}{\alpha(\G)} \leq 1$ (as any independent set in an induced subgraph is also an independent set in the full graph):

\begin{align*}
\frac{n - \FCP(\G)}{\Ind(\G)} &\leq \frac{\frac{l-2}{2l-2} \cdot n - \frac{l-4}{2l-2} \cdot t + \frac{l}{2l-2} \cdot \alpha(\G')}{\alpha(\G)} \\
&=  \frac{\frac{l-2}{2l-2} \cdot n - \frac{l-4}{2l-2} \cdot t}{\alpha(\G)} + \frac{\frac{l}{2l-2} \cdot \alpha(\G')}{\alpha(\G)} \\
&\leq \frac{\frac{l-2}{2l-2} \cdot n - \frac{l-4}{2l-2} \cdot t}{\alpha(\G)} + \frac{l}{2l-2} \\
&\leq \frac{\frac{l-2}{2l-2} \cdot n - \frac{l-4}{2l-2} \cdot t}{n/j} + \frac{l}{2l-2} \\
&= j \cdot \frac{l-2}{2l-2} - j \cdot \frac{l-4}{2l-2} \cdot \frac{t}{n} + \frac{l}{2l-2}.
\end{align*}
\end{proof}

If instead $\chi(\G') = 0$ or $1$, we have $\FCP(\G) \geq 2t + k$, so $n - \FCP(\G) \leq n - 2t - (n-3t-\alpha(\G')) = t + \alpha(\G')$, and thus $\frac{n - \FCP(\G)}{\Ind(\G)} \leq \frac{tj}{n} + 1 \leq \frac{j}{3} + 1$ using the notation above. One interesting feature of this bound is that the second term is negative for $l < 4$, but positive for $l > 4$, meaning that if $\chi(\G') = 2$ or $3$, then the bound is better when $\G$ has less triangles, but for $\chi(\G') > 4$ the bound becomes better as $\G$ has more triangles.

As an example of when this bound might be useful, consider the case where $\G$ is triangle-free outerplanar, so $\chi(\G) = \chi(\G') \leq 3$, and $t = 0$. Then we have
\begin{equation*}
\frac{n - \FCP(\G)}{\Ind(\G)} \leq 3 \cdot \frac{1}{4} - 3 \cdot \frac{-1}{4} \cdot \frac{0}{n} + \frac{3}{4} = \frac{3}{2},
\end{equation*}
so for this graph family the bound gives a $\frac{3}{2}$-approximation of $\Ind(\G)$. We will see later a result which attains approximation factor $\frac{3}{2}$ for general outerplanar $\G$ (not necessarily triangle-free), but there may be other graph families where this bound is the best available, in particular if $\chi(\G)$ and $\chi(\G')$ are both larger than 4 and $\G$ is known to contain a large set of triangles. We will use this bound later to prove a result about unit disk graphs as well.

Next, we show how to bound slightly differently in order to get a bound that does not depend on the chromatic number of $\G$, only on the number of triangles in $\G$ and the chromatic number of $\G' = \G \rvert_{V(\G) \setminus T}$.

\begin{theorem}
\label{triangle_ind_bound}
Let $\G$ be a graph, $T$ be the vertices of a maximal set of $t = \frac{|T|}{3}$ vertex-disjoint triangles in $\G$, $\G' = \G \rvert_{V(\G) \setminus T}$, and $\chi(\G') = l \geq 2$. Then
\begin{equation*}
\frac{\FCC(\G)}{\Ind(\G)} = \frac{n - \FCP(\G)}{\Ind(\G)} \leq \frac{l^2 n - 2ln - l^2 t + 4lt}{(2l-2)(n-3t)} + \frac{l}{2l-2}.
\end{equation*}
\end{theorem}
\begin{proof}
We once again use the bound
\begin{equation*}
n - \FCP(\G) \leq \frac{l-2}{2l-2} \cdot n - \frac{l-4}{2l-2} \cdot t + \frac{l}{2l-2}
\end{equation*}
from the proof of \cref{chromatic_ind_bound}, but instead of bounding $\Ind(\G) \geq \alpha(\G) \geq \frac{n}{\chi(\G)}$ as before, we bound using
\begin{equation*}
\Ind(\G) \geq \alpha(\G) \geq \alpha(\G') \geq \frac{n - 3t}{l},
\end{equation*}
which may be better when the chromatic number of $\G$ is large but not that of $\G'$, depending on the number of triangles in $\G$. This yields

\begin{align*}
\frac{n - \FCP(\G)}{\Ind(\G)} &\leq \frac{l-2}{2l-2} \cdot \frac{n}{\alpha(\G')} - \frac{l-4}{2l-2} \cdot \frac{t}{\alpha(\G')} + \frac{l}{2l-2} \\
&\leq \frac{(l-2)nl}{(2l-2)(n-3t)} - \frac{(l-4)tl}{(2l-2)(n-3t)} + \frac{l}{2l-2} \\
&= \frac{l^2 n - 2ln - l^2 t + 4lt}{(2l-2)(n-3t)} + \frac{l}{2l-2}.
\end{align*}
\end{proof}

When $\G$ is triangle-free, $\chi(\G) = \chi(\G')$ and the bounds in \cref{chromatic_ind_bound} and \cref{triangle_ind_bound} coincide. Similar also to \cref{chromatic_ind_bound}, if $\chi(\G') = 1$, one can show $\frac{n - \FCP(\G)}{\Ind(\G)} \leq \frac{t}{n - 3t} + 1$. Again, this bound will be used later to prove a result for unit disk graphs, as it is known that triangle-free unit disk graphs have small chromatic number even though unit disk graphs with triangles can have very large chromatic number.

Now we are ready to show our main result on index coding rate, which depends on the chromatic number of $\G$ and makes use of integer programming formulations of maximum independent set. To begin, we have always the lower bound
\begin{equation*}
\alpha(\G) \leq \Ind(\G),
\end{equation*}
and if $\G$ is $k$-colorable, as the largest color class is an independent set, we have
\begin{equation*}
\frac{n}{k} \leq \alpha(\G).
\end{equation*}
For an upper bound, it is shown in \cite{DBLP:journals/corr/abs-1004-1379} that
\begin{equation*}
\Ind(\G) \leq \FCC(\G).
\end{equation*}

The dual of the linear program for $\FCC(\G)$, written below, is a fractional version of maximum independent set with additional constraints for cliques of size greater than 2:
\begin{align}
\textrm{max.} \hspace{1cm} & \sum_{v \in V(\G)} x_v &\\
\textrm{s.t.} \hspace{1cm} & \sum_{v : v \in C} x_v \leq 1 & \textrm{for every clique $C$ in $\G$}  \\
& 0 \leq x_v \leq 1 & \forall v \in V(\G).
\end{align}
We denote the solution of this dual LP with all clique constraints by $\alpha_{F_n}(\G)$, and the solution of the corresponding LP with only clique constraints for cliques of size $\leq k$ by $\alpha_{F_k}(\G)$. If $k$ is a constant, then we can compute $\alpha_{F_k}(\G)$ efficiently, but we cannot compute $\alpha_{F_n}(\G)$ efficiently in general as it may have exponentially many constraints. Then as $\FCC(\G) = \alpha_{F_n}(\G)$ by duality, we have
\begin{equation*}
\alpha(\G) \leq \alpha_{F_n}(\G) = \FCC(\G) \leq \alpha_{F_{n-1}}(\G) \leq \cdots \leq \alpha_{F_2}(\G).
\end{equation*}
Since we can achieve index coding rate $\FCC(\G)$, and $\FCC(\G)$ is sandwiched between $\alpha(\G)$ and $\alpha_{F_2}(\G)$, we proceed by showing the integrality gap between these latter two quantities is fairly small for graphs with small chromatic number, from which it follows that $\FCC(\G)$ is a good approximation of $\Ind(\G)$ on these graphs. The following is a generalization of an observation made in \cite{DBLP:conf/approx/MagenM09} regarding planar graphs.

\begin{theorem}
\label{chromatic_lp_ind_bound}
If $\G$ is $k$-colorable ($k \geq 2$), then
\begin{equation*}
\frac{2}{k} \cdot \alpha_{F_2}(\G) \leq \alpha(\G) \leq \alpha_{F_2}(\G).
\end{equation*}
\end{theorem}
\begin{proof}
The upper bound is discussed above. For the lower bound, assume we have an LP solution with value $\alpha_{F_2}(\G)$. It is well-known that the linear program for independent set is $\frac{1}{2}$-integral \cite{DBLP:journals/mp/NemhauserT75}, so we can assume all $x_v$ take values in $\set{0, \frac{1}{2}, 1}$. Let $V_1$ be the set of vertices $v$ with $x_v = 1$, and $V_{1/2}$ the set with $x_v = \frac{1}{2}$, so that $\alpha_{F_2}(\G)= |V_1| + \frac{1}{2} \cdot |V_{1/2}|$. Now suppose we $k$-color the vertices of $\G$, and let $I_1, I_2, \dotsc, I_k \subseteq V_{1/2}$ be the subsets of $V_{1/2}$ corresponding to the color classes such that
\begin{equation*}
|I_1| \geq |I_2| \geq \cdots \geq |I_k|.
\end{equation*}
Now we round the fractional solution to an integral one in the following way: for every vertex $v \in I_1$, set $x_v = 1$, and for every vertex $v$ in $I_2, I_3, \dotsc, I_k$, set $x_v = 0$. This does not violate any constraints, as if in the fractional solution $x_v = 1/2$, then every neighbor $u$ of $v$ has either $x_u = 0$, or $x_u = \frac{1}{2}$, and if $v$ was rounded up it must have been in $I_1$, in which case all neighbors are in a different color class, so are rounded down.

The cost of the rounded solution is $|V_1| + |I_1|$, and as $I_1$ is the largest of the $k$ sets, we have $|I_1| \geq \frac{|V_{1/2}|}{k}$, so
\begin{equation*}
\alpha(\G) \geq |V_1| + |I_1| \geq |V_1| + \frac{|V_{1/2}|}{k} \geq \frac{2}{k} \cdot (|V_1| + \frac{1}{2} \cdot |V_{1/2}|) = \frac{2}{k} \cdot \alpha_{F_2}(\G).
\end{equation*}
\end{proof}

We note that the above bound is tight, as the all-$\frac{1}{2}$ solution is feasible for any graph, and thus an $l$-clique has $\alpha(\G) = 1, \alpha_{F_2}(\G) \geq \frac{l}{2}$. For our purposes though, improvement might be possible by instead bounding the gap between $\alpha(\G)$ and $\alpha_{F_i}(\G)$ for some $i > 2$. Some efforts in this direction and limitations to this approach are discussed in \cref{sec:further}.

\begin{corollary}
\label{chromatic_ind_cor}
Let $\G$ be a $k$-colorable graph ($k \geq 2$). Then
\begin{equation*}
\frac{\FCC(\G)}{\Ind(\G)} \leq \frac{k}{2}.
\end{equation*}
\end{corollary}
\begin{proof}
This follows immediately from the fact that $\alpha(\G) \leq \Ind(\G) \leq \FCC(\G) \leq \alpha_{F_2}(\G)$ and the previous theorem.
\end{proof}

By Brooks' theorem a graph with maximum degree $\Delta$ has chromatic number at most $\Delta + 1$, so we obtain also a result for graphs with small maximum degree.
\begin{corollary}
Let $\G$ be a graph with maximum degree $\Delta$. Then
\begin{equation*}
\frac{\FCC(\G)}{\Ind(\G)} \leq \frac{\Delta + 1}{2}.
\end{equation*}
\end{corollary}

In fact, the theorem shows that $\alpha_{F_2}(\G)$ is a $\frac{k}{2}$-approximation of $\Ind(\G)$, which may be useful in the case that $\FCC(\G)$ cannot be computed efficiently (such as if the graph family contains arbitrarily large cliques). When nothing is known about the number of triangles in $\G$ or the chromatic number of $\G \rvert_{V(\G) \setminus T}$ other than the trivial bounds, then the bound in \cref{chromatic_ind_cor} is a strict improvement over the bounds in \cref{chromatic_ind_bound} and \cref{triangle_ind_bound}.

\subsection{Bounds Based on Graph Sparsity}
When the graph is known to be sparse, Tur\'an's theorem guarantees the existence of a large independent set. If such a set is large enough, the fact that $\alpha(\G) \leq \Ind(\G) \leq \FCC(\G) \leq n$ may give a nontrivial approximation.

\begin{theorem}
Let $\G$ be a graph with $n$ vertices and $e$ edges, $T$ the vertices of a maximal set of vertex disjoint triangles, and $\G' = \G \rvert_{V(\G) \setminus T}$. If $\chi(\G') = l$ and $l > 3$, then
\begin{equation*}
\frac{\FCC(\G)}{\Ind(\G)} = \frac{n - \FCP(\G)}{\Ind(\G)} \leq \frac{l-2}{2l-2} \cdot \frac{2e}{n} + 1.
\end{equation*}
\end{theorem}
\begin{proof}
Tur\'an's theorem tells us that
\begin{equation*}
\frac{n}{\frac{2e}{n}+1} \leq \alpha(\G) \leq \Ind(\G),
\end{equation*}
and recall from the proof of \cref{chromatic_ind_bound} that
\begin{equation*}
n - \FCP(\G) \leq \frac{l-2}{2l-2} \cdot n - \frac{l-4}{2l-2} \cdot t + \frac{l}{2l-2} \cdot \alpha(\G') \leq \frac{l-2}{2l-2} \cdot n + \frac{l}{2l-2} \cdot \alpha(\G'),
\end{equation*}
where $t = \frac{|T|}{3}$, assuming $l > 3$.
Combining, we have
\begin{align*}
\frac{n - \FCP(\G)}{\Ind(\G)} &\leq \frac{l-2}{2l-2} \cdot \frac{n}{\alpha(\G)} + \frac{l}{2l-2} \cdot \frac{\alpha(\G')}{\alpha(\G)} \\
&\leq \frac{l-2}{2l-2} \cdot \frac{n}{\alpha(\G)} + \frac{l}{2l-2} \\
&\leq \frac{l-2}{2l-2} \cdot \frac{n(\frac{2e}{n}+1)}{n} + \frac{l}{2l-2} \\
&= \frac{l-2}{2l-2} \cdot \left( \frac{2e}{n}+1 \right) + \frac{l}{2l-2} \\
&= \frac{l-2}{2l-2} \cdot \frac{2e}{n} + 1.
\end{align*}
\end{proof}

If instead $\chi(\G') \leq 3$, we cannot bound in exactly the same way (we can no longer upper bound the term $-\frac{l-4}{2l-2} \cdot t$ by 0), but can use essentially the same techniques to recover the bounds:
\begin{align*}
\chi(\G')=1, \chi(\G')=2 &\implies \frac{n - \FCP(\G)}{\Ind(\G)} \leq \frac{2et}{n^2} + \frac{t}{n} + 1 \leq \frac{2e}{3n} + \frac{4}{3}, \\
\chi(\G')=3 &\implies \frac{n - \FCP(\G)}{\Ind(\G)} \leq \frac{2et}{4n^2} + \frac{2e+t}{4n} + 1 \leq \frac{2e}{3n} + \frac{13}{12}.
\end{align*}

\begin{corollary}
\label{sparse_corollary}
Let $\G$ be a graph with $n$ vertices and $e$ edges. Then
\begin{equation*}
\frac{\FCC(\G)}{\Ind(\G)} = \frac{n - \FCP(\G)}{\Ind(\G)} \leq \max\left(\frac{e(n-2)}{n(n-1)} + 1, \frac{2e}{3n} + \frac{4}{3}\right).
\end{equation*}
\end{corollary}

We note that a similar result to \cref{sparse_corollary} in the more general context of directed graphs appears in \cite{DBLP:journals/corr/YiC17}, though when considering only undirected graphs our bound is slightly better.

\subsection{Bounds for Disk Graphs}

In general, the chromatic-number-based bounds proved earlier are not as useful for approximating the index coding rate of a disk graph, as such graphs can contain cliques of arbitrary size (and thus have arbitrarily large chromatic number). However, the situation for approximating storage capacity is much better, as even for general (i.e. non-unit) disk graphs, we can improve the trivial 2-approximation to a $\frac{3}{2}$-approximation. To do so, we combine a result of \cite{DBLP:journals/networks/MalesiskaPW98} showing that every triangle-free disk graph is 3-colorable with \cref{chromatic_cap_bound}, which depends only on the chromatic number of $\G \rvert_{V(\G) \setminus T}$.

\begin{corollary}
\label{cor:disk_cap}
Let $\G$ be a disk graph, $T$ be the vertices of a maximal set of $t = \frac{|T|}{3}$ vertex-disjoint triangles, and $k$ be the size of a minimum vertex cover of $\G \rvert_{V(\G) \setminus T}$. Then
\begin{equation*}
\frac{\Cap(\G)}{\FCP(\G)} \leq \frac{3t + k}{2t + 3k/4} \leq \frac{3}{2}.
\end{equation*}
\end{corollary}

Note that without other assumptions on the graph, we may not be able to efficiently compute $\FCP(\G)$ if $\G$ has a superpolynomial number of cliques. We could attempt to instead use the weaker approximation from the proof of \cref{chromatic_cap_bound} which is used to prove the bound on $\FCP(\G)$ originally, but this requires finding a 3-coloring of $\G \rvert_{V(\G) \setminus T}$, which is hard even when the graph is known to be 3-colorable \cite{DBLP:journals/combinatorica/KhannaLS00}.

To approximate the index coding rate, we have a similar situation; we can show $\FCC(\G)$ is a good approximation by combining several known results, but we may not be able to efficiently compute $\FCC(\G)$ without imposing some further restrictions on $\G$. We first show $\FCC(\G)$ is a good approximation by combining the following two results, the first from \cite{peeters1991coloring} and the second from \cite{DBLP:journals/endm/ChalermsookV16}:

\begin{theorem}[Peeters 1991]
\label{peeters_thm}
If $\G$ is a unit disk graph, $\chi(\G) \leq 3 \omega(\G) - 2$.
\end{theorem}

\begin{theorem}[Chalermsook and Vaz 2017]
\label{clique_replacement_thm}
Let $\mathcal{F}$ be a graph family closed under clique-replacement (replacement of a vertex by a clique of arbitrary size). If there exists a constant $c$ such that for every graph $\G \in \mathcal{F}$, $\chi(\G) \leq c \cdot \omega(\G)$, then $\FCC(\G) \leq c \cdot \alpha(\G)$.
\end{theorem}

Unit disk graphs are closed under clique-replacement, as we can just replace the circle corresponding to the vertex in the geometric representation by $k$ circles in the same location, and the resulting graph will have the single vertex replaced by a $k$-clique. Then since \cref{peeters_thm} shows $\chi(\G) \leq 3 \omega(\G)$ for any UDG $\G$, we can apply \cref{clique_replacement_thm} and get that
\begin{equation*}
\alpha(\G) \leq \Ind(\G) \leq \FCC(\G) \leq 3 \alpha(\G),
\end{equation*}
yielding the following result.

\begin{theorem}
Let $\G$ be a unit disk graph. Then
\begin{equation*}
\frac{\FCC(\G)}{\Ind(\G)} \leq 3.
\end{equation*}
\end{theorem}

As mentioned above, if we want the runtime of the approximation to be polynomial, we need to impose some additional restrictions on $\G$. In the special case that the graph contains no large clique, we can combine the bound of \cref{chromatic_ind_bound} with \cref{peeters_thm} to get an approximation of $\Ind(\G)$ in terms of the clique number.

\begin{theorem}
\label{clique_udg_bound}
If $\G$ is a unit disk graph with clique number at most $\omega(\G)$, then
\begin{equation*}
\frac{n - \FCP(\G)}{\Ind(\G)} \leq \omega(\G) + 1.
\end{equation*}
\end{theorem}
\begin{proof}
As before, let $T$ be a maximal set of vertex-disjoint triangles with $|T| = t$, and $\G' = \G \rvert_{V(\G) \setminus T}$. There are several cases depending on $\chi(\G') = l$. As triangle-free disk graphs are 3-colorable, we know that $l \leq 3$. If $l = 0$ or $l=1$, then plugging in from \cref{chromatic_ind_bound} we have
\begin{equation*}
\frac{n - \FCP(\G)}{\Ind(\G)} \leq \frac{\chi(\G)}{3} + 1 \leq \frac{3 \omega(\G) - 2}{3} + 1 = \omega(\G) + \frac{1}{3}.
\end{equation*}
If $l = 2$, we have
\begin{equation*}
\frac{n - \FCP(\G)}{\Ind(\G)} \leq (3 \omega(\G) - 2) \cdot \frac{t}{n} + 1 = 3 \omega(\G) \cdot \frac{t}{n} - 2 \cdot \frac{t}{n} + 1 \leq \omega(\G) + 1.
\end{equation*}
Finally, if $l = 3$ we have
\begin{equation*}
\frac{n - \FCP(\G)}{\Ind(\G)} \leq \frac{3 \omega(\G) - 2}{4} + \frac{3 \omega(\G) - 2}{4} \cdot \frac{t}{n} + \frac{3}{4} \leq \omega(\G) + \frac{1}{12}.
\end{equation*}
\end{proof}

Recall that $\lambda$-precision unit disk graphs are unit disk graphs with the extra constraint that in the geometric representation, every pair of disk centers are distance at least $\lambda$ from one another. We can use a geometric argument to translate this constraint into a bound on the size of the largest clique, and then apply the previous theorem to obtain an approximation of $\Ind(\G)$ for this graph family.

\begin{theorem}
Let $\G$ be a $\lambda$-precision unit disk graph. Then $\omega(\G) \leq \frac{64}{\lambda^2}$.
\end{theorem}
\begin{proof}
Suppose $\G$ contains a $k$-clique. Then by definition, the geometric representation must contain a set of $k$ pairwise intersecting unit disks. We claim that regardless of $k$, these disks can all be inscribed in a circle of radius at most 4. Clearly if $k = 2$, a circle of radius 2 suffices. When $k = 3$, the worst case is that each pair of circles touches at a single point, in which case Descartes' circle theorem tells us that the circle inscribing them has radius $1 + \frac{2\sqrt{3}}{3} < 4$.

Now suppose we already have three pairwise intersecting circles of radius 1, and consider the possible locations for a fourth circle of radius 1 which intersects all three circles pairwise. It must be the case that any point on the fourth circle is distance at most 4 from any point on any of the first three circles, otherwise they could not intersect. To say the fourth circle intersects the first circle is equivalent to saying that if we draw a circle of radius 2 centered at the center of the first circle, it must contain the center of the fourth circle. The same is true for the second and third circles, so the fourth circle center must lie within the intersection of three circles of radius 2 drawn centered on the first three circles. Any point in this intersection is distance at most 3 from any point on any of the first three circles (as the greatest distance between any point in the circle of radius 2 and the circle of radius 1 centered at the same point is 3), so if we draw a circle of radius 1 centered within this intersection, every point on that circle will be distance at most 4 to any point on any small circle. Thus if we draw a circle of radius 4 centered at any point of any small circle, it will contain not only all three original circles, but also every possible location for every fourth circle. Adding a fourth circle only decreases the set of possible locations for a fifth circle and so on, so this circle of radius 4 will in fact contain all $k$ circles for any $k$.

Now, as the unit disks are $\lambda$-precision, we can think of a smaller disk of radius $\frac{\lambda}{2}$ around the center of each circle, and it must be the case that any two such disks are disjoint (except possibly sharing a single point), otherwise the two unit disk centers would be at distance $< \lambda$ from one another. Thus since all unit disks in the same clique lie in a circle of radius at most 4, we can bound the size of the maximum clique by counting how many disks of radius $\frac{\lambda}{2}$ can be packed within such a circle.

The large circle has area $16 \pi$, and the small circles each have area $\pi \cdot \frac{\lambda^2}{4}$, so there can be at most
\begin{equation*}
16 \pi / (\pi \cdot \frac{\lambda^2}{4}) = \frac{64}{\lambda^2}
\end{equation*}
small circles packed within the large circle, and all such small circles must lie entirely within the large circle because they each have radius $\lambda/2 \leq 1$, and are each centered on a unit disk which lies within the large circle by construction. Thus this is an upper bound on the size of the largest clique in $\G$.
\end{proof}

\begin{corollary}
Let $\G$ be a $\lambda$-precision unit disk graph. Then
\begin{equation*}
\frac{\FCC(\G)}{\Ind(\G)} = \frac{n - \FCP(\G)}{\Ind(\G)} \leq \frac{64}{\lambda^2} + 1.
\end{equation*}
\end{corollary}

It seems likely that the coefficient of $\lambda^{-2}$ could be made much smaller, by showing any $k$ pairwise intersecting unit disks can be inscribed in a circle of radius $<4$. Intuitively it seems a circle of radius $1 + \frac{2\sqrt{3}}{3}$  should suffice in the case of $k > 3$ circles just as it does for 3 circles, which would reduce the constant from $64$ to about $18.6$, but a more sophisticated geometric argument is needed.

It is shown in \cite{DBLP:journals/tcs/DiazK07} that for $\lambda > 1/\sqrt{2} \approx 0.707$, every $\lambda$-precision unit disk graph is planar, in which case \cite{DBLP:journals/corr/Mazumdar0V17} gives a 2-approximation of $\Ind(G)$, a significant improvement over the previous theorem. But the previous result is relevant for $\lambda \leq 1/\sqrt{2}$, where UDGs are not known to fall into any other easy-to-approximate graph family.

\section{Index Code Constructions for General Graphs}
\label{sec:IC_constructions}

In this section we provide the index code constructions for the schemes presented in \cref{sec:algo}.

\subsection{Achievability Scheme (Proof of \cref{thm:main})}\label{sec:achiev_sch}

We first describe an index coding scheme that achieves a broadcast rate equal to the optimal solution of the integer program version of the linear program in \cref{local_partial_IP}.

Assume without loss of generality  that $\cS_1 = [n_1],\; \cS_2 = [n_1+1,n_1+n_2],\; \ldots, \cS_t = [\sum_{j\in[t-1]}n_j+1,\sum_{j\in[t]}n_j]$ are the partial cliques selected. Let $k_j \triangleq k_{\cS_j}$. Assume that the optimum value of the integer program is $m$. Then $\max_j (k_j+1) \leq m \leq \sum_j (k_j+1)$. Let $k^j = \sum_{l=1}^j (k_l+1)$. Let $\Phi$ be a $[k^t,m]$-MDS matrix, such that $\Phi_j \triangleq \Phi_{[k^j+1,k^{j+1}]}$ represent submatrices of $\Phi$, and let $G_j$, $j \in [t]$ be $t$ distinct $[n_j,k_j+1]$-MDS matrices. Let 
\begin{equation}
\label{code_construction}
	[\bfu_1 \;\bfu_2 \;\dots \; \bfu_n]_{m\times n} \triangleq [\Phi_1 G_1 \; \Phi_2 G_2 \; \ldots \; \Phi_t G_t],	
\end{equation}
so that we assign vector $\bfu_i$ to vertex $v_i, i\in [n]$.

Without loss of generality consider a vertex $v_i$ in graph $\cG$ such that $i\in \cS_1$. Let 
\begin{equation}
P_j \triangleq \br{v_i\cup N(v_i,\overline\cG)} \cap \cal{S}_j, \;j\in [t]
\end{equation}
denote the data unknown to vertex $v_i$ in each of the selected partial cliques. Note that $P_1 \leq k_1+1$ and for any $j\geq 2$ such that $\abs{P_j} \geq k_j+1$, there exists a set of $P^\prime_j \subseteq P_j$ with the property that $\abs{P^\prime_j} = k_j+1$ and $\spn\br*{\bf{u}_{P_j}} = \spn\br*{\bf{u}_{P^\prime_j}}$. Let $P^\prime_j = P_j$ for $j$ with $\abs{P_j} \leq k_j+1$ and $P^\prime = \bigcup_j P^\prime_j$. If the vectors $\bfu_{P^\prime}$ are independent, then it is easy to see that $\bfu_i \not\in \spn\br*{\set{\bfu_j}_{j \in N(v_i,\overline \cG)}}  = \spn\br*{\bfu_{P\setminus i}}= \spn\br*{\bfu_{P^\prime\setminus i}}$. \Cref{lem:full_rank} shows that this is indeed that case, i.e. there exist constructions of matrices $\Phi$ and $G_j$ such that the vectors $\bfu_i$ satisfy the interference alignment criteria, $\bfu_i \not\in \spn\br*{\set{\bfu_j}_{j \in N(v_i,\overline \cG)}}$.

{\lemma \label{lem:full_rank} For any set of interfering nodes $\set{P_j}_{j\in [t]}$, there exist constructions of matrices $\Phi$ and $G_j$ over a field of size $O(n)$, such that the vectors $\bfu_{P^\prime}$ are independent.}

\begin{proof}
Let ${\tilde{G}}_j$ be any $(k_j+1)\times (k_j+1)$ submatrix  of $G_j$. Let the element in row $p$ and column $q$ of $G_j$ be $\mbox{$\tilde{G}_j$}_{p,q} = {\alpha_{j,p}}^{q-1}$, where $\alpha_{j,p}$ are non-zero elements in a field $\F_q$. First, we show that for a large enough field $\F_q$ there exist constructions of matrices $\Phi$ and $G_j$ such that $\hat{G}=[\Phi_1 {\tilde{G}}_1\;\;\Phi_2 {\tilde{G}}_2\;\;\cdots\;\;\Phi_t {\tilde{G}}_t]$ is an MDS matrix. Since $\bfu_{P^\prime}$ are a subset of the column vectors in $\hat{G}$ and $\sum_{j\in [t]} P^\prime_j \leq m$ by construction, the vectors $\bfu_{P^\prime}$ must be independent.

Let $\Phi = [\bf{v}_1\;\bf{v}_2\;\cdots\;\bf{v}_{k^t}]$ and let $\tilde\Phi_{\geq 2,s}$ denote any $m\times s$ sub-matrix of $\Phi_{\geq2} \triangleq [\Phi_2\;\cdots\;\Phi_t]$. Since $\Phi$ is MDS, $[\bf{v}_{i_1} \;\bf{v}_{i_2}\;\bf{v}_{i_r} \;\tilde\Phi_{\geq2,m-r}]$ must be full rank for all $\set{i_1, \ldots, i_r}\subseteq [k_1+1]$. Without loss of generality let $\set{i_1, i_2, \ldots, i_r}=[r]$.

For $\bfa \in \F_q^m$, consider the vector $\bf{w}\in\F_q^m$,
\begin{equation}\label{w_a}
	\bf{w} = [\bfv_1 \;\bfv_2\;\bfv_r \;\tilde\Phi_{\geq2,m-r}] \bfa
\end{equation}
such that 
$\bfa  = \begin{bmatrix}
    \bf{a}_{[r]}\\
    \bf{a}_{[r+1,m]}
\end{bmatrix}$
for $ \bfa_{[r]} \in \F_q^r$ and $\bf{a}_{[r+1,m]}\in \F_q^{m-r}$. We show that for any $\bfa \in \F_q^m$ there exist $\alpha_{1,i} \in \F_q, i\in [n_1]$ such that $\bf{w}$ can also be represented as a linear combination of column vectors in $G^\prime \triangleq [\Phi_1 H_1 \;\;\tilde\Phi_{\geq2,m-r}]$ where $H_1$ is a $(k_1+1)\times r$ submatrix of $\tilde G_1$ such that
  \begin{equation*}
H_1 = \begin{bmatrix}
  1 & 1 & \cdots &1  \\ 
  \alpha_{1,1} & \alpha_{1,2} & \cdots &\alpha_{1,r}  \\ 
  && \vdots \\
  {\alpha_{1,1}}^{k_1} & {\alpha_{1,2}}^{k_1} & \cdots &{\alpha_{1,r}}^{k_1}  \\ 
\end{bmatrix}.
\end{equation*}
We want to prove that for any $\bfa\in \F_q^m$ there exists $\bfd = \begin{bmatrix} \bf{d}_{[r]} \\ \bf{d}_{[r+1,m]} \end{bmatrix}$ such that
\begin{equation}\label{w_as_mds}
    \bf{w} = G^\prime\; \bfd
\end{equation}
for some $\bf{d}_{[r]} \in \F_q^r$ and $\bf{d}_{[r+1,m]}\in \F_q^{m-r}$.

Since $[\bf{v}_{1} \;\bf{v}_{2}\;\bf{v}_{r} \;\tilde\Phi_{\geq2,m-r}]$ is full rank, there must exist unique matrices ${B}_{[r]} \in \F_q^{r\times\br{k_1+1-r}}$, ${B}_{[r+1,m]} \in \F_q^{(m-r)\times\br{k_1+1-r}}$ such that
  \begin{equation}\label{v_r_1}
    [\bf{v}_{r+1}  \; \bf{v}_{r+2} \; \ldots\; \bf{v}_{k_1+1} ] =  [\bf{v}_1 \;\bf{v}_2\;\bf{v}_r \;\tilde\Phi_{\geq2,m-r}] \begin{bmatrix} {{B}_{[r]}} \\ {{B}_{[r+1,m]}} \end{bmatrix}. 
  \end{equation}
Thus, combining \cref{w_a,w_as_mds,v_r_1}, we have
  \begin{equation}\label{final_cond}
      [\bf{v}_1\; \bf{v}_2\;\cdots\;\bf{v}_r\; \tilde\Phi_{\geq2,m-r}]\cdot  \br*{\begin{bmatrix}
              [I_r\; B_{[r]}] \; H_1 \bfd_{[r]} \\
              B_{[r+1,m]} \tilde H_1 \bfd_{[r]} + \bf{d}_{[r+1,m]}
            \end{bmatrix} - \begin{bmatrix}
              \bf{a}_{[r]}\\
              \bf{a}_{[r+1,m]}
            \end{bmatrix}} = \bf{0},
  \end{equation}
  where
  \begin{equation*}
    \tilde H_1 = \begin{bmatrix}
      {\alpha_{1,1}}^{r+1} & {\alpha_{1,2}}^{r+1} & \cdots &{\alpha_{1,r}}^{r+1}  \\ 
      {\alpha_{1,1}}^{r+2} & {\alpha_{1,2}}^{r+2} & \cdots &{\alpha_{1,r}}^{r+2}  \\ 
      && \vdots \\
      {\alpha_{1,1}}^{k_1} & {\alpha_{1,2}}^{k_1} & \cdots &{\alpha_{1,r}}^{k_1}  \\ 
    \end{bmatrix}
  \end{equation*}
and $I_r$ denotes the $r\times r$ identity matrix. For the solution in \cref{final_cond} to exist for all $\bf{a}_r \in \F_q^r, \bf{a}_{r+1} \in \F_q^{k-r+1}$ we must have
$
    \det \br*{[I_r\; B_r] H_1} \ne 0,
$
  or equivalently
  \begin{align} \label{field_size_condition}
  \det\left[
\begin{matrix}
     \bf{g}({\alpha_{1,1}})\bf{b}_1 & \bf{g}({\alpha_{1,2}})\bf{b}_1 &\cdots    &      \bf{g}({\alpha_{1,r}})\bf{b}_1\\
     \alpha_{1,1}+\bf{g}({\alpha_{1,1}})\bf{b}_2 & \alpha_{1,2}+\bf{g}({\alpha_{1,2}})\bf{b}_2 &\cdots   & \alpha_{1,r}+\bf{g}({\alpha_{1,r}})\bf{b}_2\\
     &\vdots&\vdots & \vdots\\
     {\alpha_{1,1}}^{r-1}+\bf{g}({\alpha_{1,1}})\bf{b}_r & {\alpha_{1,2}}^{r-1}+\bf{g}({\alpha_{1,2}})\bf{b}_r &\cdots    & {\alpha_{1,r}}^{r-1}+\bf{g}({\alpha_{1,r}})\bf{b}_r\\
   \end{matrix} \right]
    \ne0
  \end{align}
  where $\bf{g}(\alpha)=[\alpha^{r} \;\alpha^{r+1}\; \cdots\; \alpha^{k_1}]$ and $B_{[r]} = \begin{bmatrix}
    \bf{b}_1\,\,
    \bf{b}_2 \,\,
    \cdots\,\,
    \bf{b}_r
  \end{bmatrix}^T.$

If we expand out the polynomial, the determinant in the left hand side of \cref{field_size_condition} has degree at most $k_1$ in each of the variables $\alpha_{1,1}, \alpha_{1,2},\ldots, \alpha_{1,r}$. Thus, by increasing the size of the field $\F_q$ we can make sure that there exist $\alpha_{j,i}$ for all $j\in [t]$ and $i \in [n_j]$ so that \cref{field_size_condition} holds for all submatrices $\tilde G_j$ and $\tilde H_j$. 

Now, repeating the above argument $t$ times we can say that $\left[\Phi_1 H_1 \;\Phi_2 H_2\;\Phi_3 H_3\;\cdots \;\Phi_t H_t\right] $ is MDS for all sets of submatrices $H_j \in \F_{q^r}^{(k_j+1)\times (k_j+1)}$ of $G_j$.

A loose upper bound on the (sufficient) field size is
\begin{equation}\label{loose_upper_bound}
 q \leq \max_{j\in[t]} \sum_{r=1}^{k_j+1} k_j {{n-n_j}\choose {m-r}} {{n_j}\choose {r-1}} +n_j.
\end{equation}

Note that in the above proof we do not need the matrix $\hat{G}$ to be MDS. Instead, we need only $n$ different subsets of column vectors of $G$ each of size at most $m$ to be linearly independent. Thus the upper bound on the size of the alphabet in \cref{loose_upper_bound} is very loose and it can be shown that an alphabet of size $O(n)$ suffices.
\end{proof}

To find the vector linear index code corresponding to the linear program in \cref{local_partial_IP}, we can modify the solution described above as follows.

Consider the optimal solution  $\rho_\cS^\star$ for the linear program in \cref{local_partial_IP}. Since all the coefficients of the  linear program in \cref{local_partial_IP} are integers, $\rho_\cS^\star$ must be rational. Assume that in the optimal solution to \cref{local_partial_IP} the partial cliques $\cS$ for which $\rho_\cS>0$ are $\cS_1,\; \cS_2,\; \ldots, \cS_t$, and $\abs{\cS_j}=n_j$. Let $k_j \triangleq k_{\cS_j}$ and $\rho_{S_j} = N_j/N$, for $N_j, N \in \mathbb{Z}^+$. Note that $\sum_{j=1}^t n_j N_j = N n$. Assume that the linear program gives an index coding rate $m^\star$.

Let $G_j, j\in [t]$ be an $[n_j,k_j+1]$-MDS matrix, and let $G_j^\prime \triangleq \sum_{l=1}^{N_j} Q_{ll}\otimes G_{j},$ where $Q_{rs}$ is an $N_j \times N_j$ matrix with the only nonzero entry being $Q_{rs}(r,s)=1$, and $\otimes$ denotes the matrix tensor product. Let $k^j \triangleq \sum_{l=1}^j N_j (k_l+1)$ and let $\Phi$ be a $[k^t,m^\star N]$-MDS matrix. Let $\Phi_j \triangleq \Phi_{\brac*{k^{j-1}+1,k^{j}}}$. Construct $[\Phi_1 G_1^\prime\; \Phi_2 G_2^\prime\; \cdots \Phi_t G_t^\prime]_{(N m^\star) \times (N n)}$ such that
\begin{equation}\label{code_construction_frac}
  [\bf{u}_{i,1}\cdots\bf{u}_{i,(n_i N_i)}] = \Phi_i G_i^\prime
\end{equation}
and assign $N_i$ vectors from $[\bf{u}_{i,1}\cdots\bf{u}_{i,(n_i N_i)}]$ to each of the vertices in $S_i$. Note that since each vertex  $v_i$ must satisfy $\sum_{S_j : S_j \ni v_i} N_j = N$, we are assigning $N$ vectors to each vertex. The interference alignment condition corresponding to vector linear index coding is similar to \cref{interference_align_condition}. In this case, since we assign multiple vectors to each vertex, we have the extra requirement that all vectors corresponding to each vertex must be independent, and that each vector assigned to a vertex is independent of all the vectors assigned to that vertex's non-neighbors. Denote by ${\bfv_{i,j}}_{j\in [N]}$ all the vectors assigned to vertex $v_i \in V(\cG)$. Therefore we have the following condition,
$$ \bfv_{i,p} \not\in \spn\br{\set{\bfv_{j,q}}_{q\in [N],j\in N(v_i,\overline\cG)} \cup \set{\bfv_{i,q}}_{q\in [N]\setminus p} }. $$
The argument that the aforementioned vector assignment satisfies this condition is similar to the argument in \cref{lem:full_rank}.

To achieve broadcast rate equal to the solution of the recursive linear program of \cref{local_partial_rec_lin_prog}, we can recursively use the scheme proposed above. More specifically, suppose that the matrices $G_1, G_2, \ldots, G_t$ represent the vector assignment satisfying the interference alignment criteria for subgraphs $\cG\vert_{\cS_1}, \cG\vert_{\cS_2}, \ldots, \cG\vert_{\cS_t}$; that is, column vectors of $G_j \in \F^{K_j \times \Theta(n_j)}$ are assigned to vertices in $\cG\vert_{\cS_j}$ corresponding to the linear program $IC_{FLP}(\cG\vert_{\cS_j})$. Let $\cS_1, \ldots, \cS_t$ be the selected subgraphs with positive weight $\rho_{\cS_j} = N_j / N$, $m^\star$ be the optimal index coding rate corresponding to the linear program $IC_{FLP}(\cG)$, and $\Phi$ be a $[\sum_j N_j K_j, m^\star N]$-MDS matrix such that $\Phi_j \triangleq \Phi_{[\sum_{l=1}^{j-1}N_j K_j+1, \sum_{l=1}^{j}N_j K_j]}$. Then the vector assignment for the graph $\cG$ would correspond to the column vectors in $[\Phi_1 G_1^\prime\; \Phi_2 G_2^\prime\; \ldots \Phi_t G_t^\prime]_{(N m^\star) \times (N n)}$, where $G_j^\prime \triangleq \sum_{l=1}^{N_j} Q_{ll}\otimes G_{j}$.

\begin{remark*}	[Codes with small alphabet size]
We note that instead of using a $[m,k+1]$-MDS matrix, the parity check matrix of any linear code of size $m$ and minimum distance $k+2$ would work. Thus, when restricted to using a small alphabet size (say $q$), we have the following upper bound on the size of the code using the Gilbert-Varshamov bound:
\begin{equation*}
m-\log_q A_q(m,\chi_l+2) = m-\log_q \br*{\frac{q^m}{\sum_{j=0}^{\chi_l+1} \binom{m}{j}(q-1)^j}} = \log_q \br*{\sum_{j=0}^{\chi_l+1} \binom{m}{j}(q-1)^j}.
\end{equation*}
\end{remark*}

\subsection{Index Code for $(k,n_1)$-GIC}\label{sec:gic}

In this section, we describe an index coding scheme based on a covering by the above type of graph, but first we present an important property of a $(k,n_1)$-GIC that allows us to construct such a scheme.

\begin{lemma} \label{lem:subtree} If a vertex $v \in V\setminus V_\rmI$ belongs to trees $T_i$ and $T_j$, $i\ne j$, then all the non-inner nodes on the subtree of $T_i$ rooted at $v$ also belong to $T_j$.
\end{lemma}

\begin{proof}

Denote the leaves of the subtree of the tree $T$ rooted at vertex $v \in V(T)$ as $L(v,T)$, and the leaves of the tree $T$ as $L(T)$. We prove the above claim in the following three lemmas.
{\lemma \label{subtree_lemma_prelim} If a vertex $v\not\in V_\rmI$ is such that $v\in V(T_i) \cap V(T_j),\; i\ne j$, then  $L(v,T_j), L(v,T_i)\subseteq V_\rmI \setminus \set{v_i,v_j}$.}
\begin{proof}
  Suppose that the vertex $ v_j\in L (v,T_i)$; then there exists a path from vertex $ v $ to $ v_j $ in the tree $ T_i $. However, in the tree $ T_j $, there is a path from vertex $ v_j $ to $ v $. Thus in the sub-digraph $ D $, there is a path from vertex $ v $ to $ v_j $ (via $ T_i $) and vice versa (via $ T_j $). As a result, there is a cycle in $D$ containing only the vertex $v_j$, contradicting $a)$ in \cref{Digraph_properties}. Hence $ v_j\notin L (v,T_i) $. In other words, $ L (v,T_i) \subseteq V_\rmI \setminus \{v_i,v_j\} $. Similarly, $ L (v,T_j) \subseteq V_\rmI \setminus \{v_i,v_j\} $. 
\end{proof}

{\lemma \label{subtree_lemma1} If a vertex $v\not\in V_\rmI$ is such that $v\in V(T_i) \cap V(T_j),\; i\ne j$, then  $L(v,T_i)=L(v,T_j)$.}
\begin{proof}
  From \cref{subtree_lemma_prelim}, $ L(v,T_i) $ is a subset of $ V_\rmI \setminus \{v_i,v_j\} $. Now pick a vertex $ v_c $ belonging to $ V_\rmI \setminus \{v_i,v_j\} $ such that $ v_c \in L(v,T_i)$ but $ v_c\notin L(v,T_j) $ (such a vertex exists since we suppose that $ L(v,T_i)\neq L(v,T_{j}) $). In tree $ T_i $, there exists a directed path from the vertex $ v_i $ which includes the vertex $ v $, and ends at the leaf vertex $ v_c $. Denote this path by $ P_{{v_i}\rightarrow v_c}(T_i)$. 

Now, suppose that in tree $ T_j $ there exists a directed path from the vertex $ v_j $ to the leaf vertex $v_c$ which doesn't include the vertex $ v $ (since $ v_c\notin L(v,T_j) $); denote this path by $ P_{{v_j}\rightarrow v_c}(T_j) $. However, in the digraph $ D $ we can also obtain a directed path from the vertex $ v_j $ which passes through the vertex $ v $ (via $ T_j$), and ends at the leaf vertex $ v_c $ (via $ T_i $), which we denote by $ P_{{v_j}\rightarrow v_c}(D) $. The paths $ P_{{v_j}\rightarrow v_c}(T_j) $ and $ P_{{v_j}\rightarrow v_c}(D)$ are different, and do not contain any other inner vertices. This contradicts condition $b)$ in \cref{Digraph_properties}.

Therefore, there cannot exist a path in tree $T_j$ from vertex $v_j$ to $v_c$, i.e. $v_c \not\in L(T_j)$. But since the tree $T_j$ must have maximum number of leaves in $V_\rmI$ and there exists a tree rooted at $v_j$ that has more leaves than $T_j$, this leads to a contradiction as well.
\end{proof}

{\lemma \label{subtree_lemma2} If a vertex $v\not\in V_\rmI$ is such that $v\in V(T_i) \cap V(T_j),\; i\ne j$, then the out-neighborhood of the vertex $v$ must be the same in both the trees, i.e. $N(v,T_i)=N(v,T_j)$.}
\begin{proof}
  Now we pick a vertex $ v_b $ such that, without loss of generality, $v_b\in N (v,T_i) $ but $v_b\not\in N(v,T_j) $ (such $ v_b $ exists since we assumed that $ N(v,T_i)\neq N(v,T_j) $). There are two cases for $ v_b $, which are $ v_b\in L(v,T_i)$ (case 1), and $ v_b\notin L(v,T_i)$ (case 2). Case 1 is addressed in \cref{subtree_lemma1}. For case 2, we pick a leaf vertex $ v_d \in L(v_b,T_i)$ such that there exists a path that starts from $ v $ followed by $ v_b $, and ends at $ v_d $, i.e., $ \langle v,v_b,\dotsc,v_d \rangle $ exists in $ T_i $. A path $ \langle v_j,\dotsc,v \rangle $ must exist in $ T_j $, thus a path $ \langle v_j,\dotsc,v,v_b,\dotsc,v_d \rangle $ exists in $D$. From the first part of the proof, we have $ L(v,T_i)=L(v,T_j)$, so $ v_d\in L(v,T_j) $. Now in $ T_j $, there exists a path from $ v_j $ to $ v_d $ which includes vertex $ v $ followed by a vertex $ v_e$ such that $ v_e\in N (v,T_j)$ and $v_e\neq v_b$ (as $ v_b\notin N (v,T_j) $), and furthermore the path ends at $ v_d $. The entire path is then $ \langle v_j,\dotsc,v,v_e,\dotsc,v_d \rangle $, which is different from $ \langle v_j,\dotsc,v,v_b,\dotsc,v_d \rangle $, so there exist distinct paths from $v_j$ to $v_d$ in $D$ that do not contain any other inner vertices, violating condition $b)$ in \cref{Digraph_properties}. Consequently, $ N(v,T_i)= N(v,T_j) $.  
\end{proof}

\end{proof}

Note that although \cref{lem:subtree} is similar to \cite[lemma 3]{thapa2015generalized}, it is different in that it applies to $(k,n_1)$-GICs in contrast to \cite{thapa2015generalized} which applies only to $(0,n_1)$-GICs.

Let $[\bf{v}_1 \;\bf{v}_2\;\cdots \;\bf{v}_{n_1}]$ be a $[n_1,k+1]$-MDS matrix. Then the broadcast symbols for the index code are:
 \begin{enumerate}
\item $\displaystyle\bfw_\rmI = \sum_{i: v_i \in V_\rmI} \bfv_i x_i$ .

\item $\bf{w}_j \in \F_q^{\min\set{\abs{N(v_j,D)},k+1}}$,  $\forall v_j \in V(\cG)\setminus V_\rmI$, where
    \begin{equation}\label{w_j}
      \bf{w}_j = 
                \begin{cases} 
                     \displaystyle\sum_{v_l \in N(v_j,D) \cap V_\rmI} \bf{v}_{l} (x_j+ x_l)  +  \sum_{v_l \in N(v_j,D) \setminus V_\rmI} \bf{u}_{l} (x_j + x_l) \hfill \text{ if } \abs{N(v_j,D)}\geq k+1 \\
                     \bf{1} (x_j) + \bf{x}_{N(v_j,D)} \hfill  \text{ if }\abs{N(v_j,D)}< k+1 \\
                \end{cases}
   \end{equation}
  where $\bf1\in \F_q^{\abs{N(v_j,D)}}$ denotes the all ones vector, $\bf{x}_{N(v_j,D)} \in \F_q^{\abs{N(v_j,D)}}$ denotes the input symbols corresponding to $N(v_j,D)$, and the vector $\bf{u}_{l} \in \F_q^{\min\set{\abs{N(v_j,D)},k+1}}$ is described in \cref{algo:select_u}.
\end{enumerate}

\begin{algorithm}[t]
  \KwData{trees $T_1, T_2, \ldots, T_n$}
  \KwResult{$\bf{u}_{j} \in \F_q^{\min\set{\abs{N(v_j,D)},k+1}}$  for $v_j \in V\setminus V_\rmI$ and $\bf{u}_j \in \F^{k+1}$ for $v_j \in V_\rmI$}
  $\bf{u}_{i} = \bf{v}_i$ for all $v_i \in V_{\rmI}$\\
  $S =  V\setminus V_\rmI$\\
    \While{$\abs{S}>0$}{
      Find a vertex $v_i\in S$ such that $N(v_i,D)\subseteq \overline{S}$ \\
      $\bf{u}_i = -\displaystyle\sum_{j:v_j\in N(v_i,D)}\bf{u}_{j}$\\
      $S = S\setminus \{v_i\}$
    }
  \caption{Selecting the vectors $\bf{u_j}, j\in V\setminus V_{\rm{I}}$.}
  \label{algo:select_u}
\end{algorithm}

Let us now prove that using the index coding scheme proposed above every vertex is able to decode the input symbols requested.

It is easy to see that all the non-inner vertices $v_j\in V\setminus V_{\mathrm{I}}$ can recover their data $x_j$. We show that $v_i \in V_{\rm{I}}$ can also recover $x_i$.  Define ${\bf{w}_j}^\prime$ corresponding to the transmitted vector $\bf{w}_j$ for $v_j \in V\setminus V_\rmI$ as
\begin{equation}
  {\bf{w}_j}^\prime = 
          \begin{cases} 
               \bf{w}_j \hfill &\text{if } \abs{N(v_j,D)}\geq k+1 \\
               [\bf{u}_{c_1}\; \bf{u}_{c_2}\; \cdots \;\bf{u}_{c_r}] \bf{w}_j \hfill & \text{if }\abs{N(v_j,D)}< k+1 \\
          \end{cases},
\end{equation}
where $\set{v_{c_1},v_{c_2},\ldots,v_{c_r}} = N(v_j,D)$. Denote by $T(v)$ the subtree rooted at vertex $v$ in tree $T$. For the non-inner children $v_j$ of vertex $v_i$ compute
\begin{align}
  \bf{w}(v_j) &= \sum_{v_l \in T_i(v_j)\setminus V_{\rm{I}}} \bf{w}^\prime_l \nonumber \\
  	&= \bfu_j x_j + \sum_{v_l \in T_i(v_j)\cap V_\rmI} \bfv_l x_l,
\end{align}
where the last equality follows from the construction of vectors $\bfu_l$ in \cref{algo:select_u} and \cref{lem:subtree}. Therefore, the terms in 
\begin{equation}\label{w_v_i}
	\bfw_I + \sum_{v_j \in N(v_i,D)\setminus V_\rmI} \bf{w}(v_j)
\end{equation} contain (at most) $k$ non-neighbors of vertex $v_i$ in the inner vertex set $V_\rmI$ and the terms $\sum_{v_j \in N(v_i,D)\setminus V_\rmI} \bfu_j x_j$ which are known to vertex $v_i$. Therefore, each vertex $v_i \in V_\rmI$ can compute $x_i$  from \cref{w_v_i}.

\begin{remark*}	
The local partial clique cover scheme considers the maximum number of partial cliques in the one-hop neighborhood of any vertex. We could similarly consider the maximum number of Generalized Interlinked Cycles in the neighborhood of a vertex. Such a scheme would combine all of the schemes presented in this paper.
\end{remark*}

\section{Directions for Further Research}
\label{sec:further}

While we have been able to show improvement in the approximation factors of both the storage capacity and index coding rate for some particular graph families, there are still questions remaining. We have also observed that it seems in general to be much harder to obtain good approximations of the index coding rate of a graph than its storage capacity, despite the fact that finding optimal solutions to the two problems is equivalently hard. In the most general case, the situation for index coding seems bleak -- it is not even known how to obtain an $O(n^{1-\epsilon})$ approximation for any $\epsilon > 0$, whereas a simple 2-approximation for the storage capacity is known. Any result either improving this approximation further or showing APX-hardness for $\Cap(\G)$ would be very interesting.

One of the primary difficulties in finding good approximations for index coding rate seems to be the lack of tools for analyzing more complicated coding schemes. Almost every result in this paper that gives a provable guarantee about index coding rate works simply by using $\FCC(\G)$ or something strictly weaker as our approximation, though we use many different means to bound the quality of the approximation. As we have seen, there are many better schemes than $\FCC(\G)$ available, such as the schemes presented in \cref{sec:IC_constructions}, but the greater complexity of these schemes seems to make the analysis much more difficult.

One of our results in particular seems as if it should be improvable with a more sophisticated analysis; recall that in order to show $\FCC(\G)$ is a $\frac{\chi(\G)}{2}$-approximation of $\Ind(\G)$, we demonstrate the chain of inequalities
\begin{equation*}
\frac{2}{\chi(\G)} \cdot \alpha_{F_2}(\G) \leq \alpha(\G) \leq \alpha_{F_n}(\G) = \FCC(\G) \leq \alpha_{F_{n-1}}(\G) \leq \cdots \leq \alpha_{F_2}(\G),
\end{equation*}
effectively showing $\alpha_{F_2}(\G)$ is a $\frac{\chi(\G)}{2}$-approximation, and thus $\FCC(\G)$ must be at least as good. In general, if the graph is dense, it may not be feasible to compute $\alpha_{F_n}(\G) = \FCC(\G)$, but for any fixed constant $k \leq n$ we can efficiently compute $\alpha_{F_k}(\G)$, which must still be a better approximation than $\alpha_{F_2}(\G)$. For example, if we restrict to considering outerplanar $\G$, our result tells us the integrality gap between $\alpha(\G)$ and $\alpha_{F_2}(\G)$ is at most $\frac{3}{2}$, and this is tight, as we can take $\G$ to be a triangle which has $\alpha_{F_2}(\G) = \frac{3}{2}$. If we move instead to $\alpha_{F_3}(\G)$, we gain another constraint in the LP which says the sum of the variables on any triangle must be at most 1, so clearly then the triangle has no integrality gap for $\alpha_{F_3}$. In fact, the worst gap we are aware of for any outerplanar graph using $\alpha_{F_3}$ is $\frac{5}{4}$, by taking $\G$ to be a 5-cycle, which is triangle-free and so has $\alpha_{F_2}(\G) = \alpha_{F_n}(\G) = \FCC(\G)$. So it is clear that we will not obtain a PTAS just by moving from $\alpha_{F_2}$ to $\alpha_{F_n}$ even for outerplanar graphs, but it seems very plausible that the approximation factor could be improved beyond $\frac{3}{2}$ by a more sophisticated analysis of the integrality gap here. There is nothing particularly special about outerplanar graphs either; a similar phenomenon seems to hold for other graph families as well. With planar graphs, for instance, the only obvious example attaining integrality gap 2 seems to be a 4-clique, which would have no gap if we used $\alpha_{F_4}$ as our approximation instead of $\alpha_{F_2}$.

In general this sequence of LPs, often referred to as ``maximum independent set with clique constraints,'' is well studied, and one might hope that some of this body of work could be leveraged to help approximate the index coding rate. For instance, Lov\'asz, while trying to approximate a different parameter $\Theta(\G)$, the ``Shannon capacity'' of $\G$, demonstrated a semidefinite program with solution referred to as the ``Lov\'asz theta function'' $\theta(\G)$, with the property that
\begin{equation*}
\alpha(\G) \leq \Theta(\G) \leq \theta(\G) \leq \alpha_{F_n}(\G) \leq \cdots \leq \alpha_{F_2}(\G),
\end{equation*}
and since the semidefinite program can be solved efficiently, we can actually compute $\theta(\G)$ efficiently \cite{DBLP:journals/tit/Lovasz79}. Unfortunately, it is not true in general that $\Ind(\G) \leq \theta(\G)$, so it is not obvious how to leverage these results. Another potential technique with similar issues would be to use an established LP hierarchy for strengthening LP solutions towards integral ones, such as the Sherali-Adams hierarchy, instead of strengthening the LP by moving from $\alpha_{F_2}(\G)$ to $\alpha_{F_k}(\G)$ for $k > 2$. This has worked in the past for some similar problems, such as maximum independent set on planar graphs, where the SA hierarchy yields a PTAS \cite{DBLP:conf/approx/MagenM09}. However there is a similar issue to that with the Lov\'asz theta function, where (at least for some graphs) at a certain level of the hierarchy the strengthened LP ceases to be an upper bound on $\Ind(\G)$.

Another direction considered in this paper was to investigate whether we could obtain good approximations for disk graphs or unit disk graphs, as these are often thought to be good models of certain types of real world networks where connections are based on some notion of proximity. While we were successful in improving the approximations for $\Cap(\G)$ and $\Ind(\G)$ on these types of graphs, we resorted to using approximations which may not be computable in polynomial time. For $\Cap(\G)$ we can always resort to the efficient 2-approximation instead, but for $\Ind(\G)$ no efficient constant-factor approximation is known for UDGs.

The primary methods used to get good approximations of other graph parameters for disk graphs rely on divide-and-conquer approaches, where the geometric representation is split into some number of pieces depending on how good of an approximation is needed, and some small portions of the representation which span multiple pieces are ignored. For packing problems like maximum independent set this works well, as any feasible solution on an induced subgraph remains feasible on the whole graph. Index coding is in this sense more like a covering problem though, where adding vertices to a graph causes previously feasible solutions to become infeasible. In general, understanding exactly how $\Ind(\G)$ varies when $\G$ has a small number of vertices or edges added or removed seems like a very difficult problem, which makes approximating $\Ind(\G)$ by divide-and-conquer approaches challenging. Even if we restrict the encoding functions to be linear, only some basic results in this direction are known, and if the functions are allowed to be nonlinear it seems even more difficult \cite{DBLP:conf/isit/BerlinerL11}. If one could show some slightly stronger results about how $\Ind(\G)$ changes under small changes to $\G$, it would likely be enough to attain good approximations for certain graph classes, such as general disk graphs, or graphs with bounded tree-width.

\bibliographystyle{plain}
\bibliography{references}

\end{document}